\pdfoutput=1
\documentclass[acmsmall,nonacm,final,screen,natbib=false]{acmart}

\setcopyright{acmlicensed}
\copyrightyear{2024}
\acmYear{2024}
\acmDOI{XXXXXXX.XXXXXXX}

\definecolor{gryffindor}{RGB}{220,0,1}
\definecolor{ravenclaw}{RGB}{14,26,164}

\usepackage{amsmath, amsthm, amssymb, mathrsfs, calc}
\usepackage{mathtools,thmtools}
\usepackage{url}
\usepackage{mleftright}
\usepackage{enumitem}
\usepackage{framed} 
\usepackage{microtype}
\usepackage[small, bold, full]{complexity}
\newclass{\PTIME}{PTIME}
\usepackage{bm}
\usepackage{pgfplots} \pgfplotsset{compat=1.18}
\usepackage{tabularx}
\usepackage{dsfont}
\usepackage[normalem]{ulem}

\usepackage[T1]{fontenc}
\usepackage[utf8]{inputenc} \DeclareUnicodeCharacter{2212}{-}

\usepackage{multirow}
\usepackage{rotating}

\newcommand{\red}[1]{\textcolor{gryffindor}{#1}}
\newcommand{\blue}[1]{\textcolor{ravenclaw}{#1}}

\usepackage{caption}
\DeclareCaptionType{Loop}

\newcommand{\bluebox}[1]{%
    \colorlet{currentcolor}{.}%
    {\color{ravenclaw}%
    \fbox{\color{currentcolor}#1}}%
}

\newcommand{\redbox}[1]{%
    \colorlet{currentcolor}{.}%
    {\color{gryffindor}%
    \fbox{\color{currentcolor}#1}}%
}

\usepackage{hyperref}
\hypersetup{
	final,
	colorlinks,
	linkcolor={ravenclaw},
	citecolor={slytherin},
	urlcolor={gryffindor}
}

\usepackage{aliascnt}
\usepackage[capitalise,nameinlink]{cleveref}

\renewcommand{\top}{\ensuremath\mathsf{T}}
\renewcommand{\epsilon}{\ensuremath\varepsilon}

\newcommand{\Id}{\operatorname{Id}}

\newcommand{\Z}{\mathbb{Z}}

\newcommand{\N}{\mathbb{N}}

\newcommand{\Q}{\mathbb{Q}}

\newcommand{\NN}{\mathbb{N}}
\newcommand{\ZZ}{\mathbb{Z}}
\newcommand{\QQ}{\mathbb{Q}}

\newcommand{\QQbar}{\smash{\overline{\Q}}\vphantom{\Q}}

\newcommand{\diag}[1]{\ensuremath{\operatorname{diag}(#1)}}
\DeclareMathOperator{\GL}{GL}
\newcommand{\ideal}[1]{\left\langle #1 \right\rangle}
 
\newcommand{\orbit}{\mathcal{O}}
\newcommand{\zski}{\smash{\overline{\mathcal{O}}}\vphantom{\mathcal{O}}}

\providecommand*{\eu}%
{\ensuremath{\mathrm{e}}}
\providecommand*{\iu}%
{\ensuremath{\mathrm{i}}}

\newcommand{\loopie}{\mathcal{L}}

\newcommand{\alg}{\smash{\overline{\Q}}\vphantom{\Q}}

\newcommand{\Lexp}[1]{L_{\operatorname{exp}}(#1)}

\newcommand{\textover}[3][l]{%
 \makebox[\widthof{#3}][#1]{#2}%
 }

\newcommand{\appref}[1]{\hyperref[#1]{Appendix~\ref*{#1}}}
\usepackage[plain]{algorithm}
\usepackage{algorithmic}

\makeatletter
\renewcommand\paragraph{\@startsection{paragraph}{4}{\z@}{3.25ex \@plus1ex \@minus.2ex}{-1em}{\bfseries\normalsize}}
\makeatother

\usetikzlibrary{arrows,decorations,intersections,pgfplots.fillbetween}
\usetikzlibrary{shapes.multipart,shapes.misc}
\usepgflibrary{decorations.pathreplacing}
\usetikzlibrary{arrows,topaths}
\usetikzlibrary{automata,shapes.arrows, backgrounds, fadings,intersections, positioning}
\usetikzlibrary{shadows}

\usetikzlibrary{decorations.pathmorphing}
\usetikzlibrary{decorations.shapes}
\usetikzlibrary{shapes}
\usetikzlibrary{backgrounds}
\usetikzlibrary{fit}
\usetikzlibrary{arrows}
\usetikzlibrary{calc}
\usetikzlibrary{mindmap}
\tikzset{every picture/.style={thick,>=angle 60}}
\tikzset{Grand/.style={draw,circle,minimum size=11*1.5,inner sep=0}}
\tikzset{Gmax/.style={draw,rectangle,minimum size=9*1.5,inner sep=0}}
\tikzset{Gmin/.style={draw,diamond,minimum size=9*1.5,inner sep=0}}
\tikzset{gamebad/.style={fill=red}}

\usepackage{amsmath,bm,times}
\RequirePackage[
	backend=biber,
	datamodel=acmdatamodel,
    url=false,
    doi=true,
    giveninits=true,
	style=acmnumeric,
	sortcites
]{biblatex}

\DeclareSourcemap{
  \maps[datatype=bibtex]{
    \map{
      \step[fieldset=issn, null]
    }
  }
}

\AtEveryBibitem{
 \clearlist{address}
 \clearfield{isbn}
 \clearfield{issn}
 \clearlist{location}
 \clearfield{month}
 \clearfield{series}
 
 \ifentrytype{book}{}{% 
  \clearlist{publisher}
  \clearname{editor}
 }
}

\AtEndPreamble{%
    \theoremstyle{acmplain}
    
    \newtheorem{claim}[theorem]{Claim}
    
    \theoremstyle{acmdefinition}
    
}

\setcopyright{none}
\renewcommand\footnotetextcopyrightpermission[1]{}
\begin{document}

\title[Simple Linear Loops: Algebraic Invariants and Applications]{\mbox{Simple Linear Loops:} \mbox{Algebraic Invariants and Applications}}

\author{Rida Ait El Manssour}
\email{manssour@irif.fr}
\affiliation{%
  \institution{IRIF,  CNRS, Universit\'e Paris Cit\'e}
    \city{Paris}
  \country{France}
}

\author{George Kenison}
\email{g.j.kenison@ljmu.ac.uk}
\affiliation{%
  \institution{Liverpool John Moores University}
  \city{Liverpool}
  \country{UK}
}

\author{Mahsa Shirmohammadi}
\email{mahsa@irif.fr}
\affiliation{%
  \institution{IRIF,  CNRS, Universit\'e Paris Cit\'e}
  \city{Paris}
  \country{France}
}

\author{Anton Varonka}
\email{anton.varonka@tuwien.ac.at}
\affiliation{%
  \institution{TU Wien}
  \city{Vienna}
  \country{Austria}
}

\begin{abstract}

The automatic generation of loop invariants is a fundamental challenge in software verification. While this task is undecidable in general, it is decidable for certain restricted classes of programs. This work focuses on invariant generation for (branching-free) loops with a single linear update.

 Our primary contribution is a polynomial-space algorithm that computes the strongest algebraic invariant for simple linear loops, generating all polynomial equations that hold among program variables across all reachable states. The key to achieving our complexity bounds lies in mitigating the blow-up associated with variable elimination and Gröbner basis computation, as seen in prior works (see~\cite{cyphert2024solvable,hrushovski2023strong,NPSHW2021} among others). Our procedure runs in polynomial time when the  number of program variables is fixed. 

We examine various applications of our results on invariant generation, focusing on invariant verification and loop synthesis. The invariant verification problem investigates whether a polynomial ideal defining  an algebraic set  serves as an invariant for a given linear loop. We show that this problem is {\coNP}-complete and lies in {\PSPACE} when the input ideal is given in dense or sparse representations, respectively.
In the context of  loop synthesis, we aim to construct a loop with an infinite set of reachable states that upholds a specified algebraic property as an invariant. The strong synthesis variant of this problem requires the construction of loops for which the given property is the strongest invariant. In terms of hardness, synthesising loops over integers (or rationals) is as hard as Hilbert's Tenth problem (or its analogue over the rationals). 
When the constants of the output are  constrained to bit-bounded rational numbers,
we demonstrate that loop synthesis and its strong variant are both  decidable in {\PSPACE}, and in {\NP} when  the  number of program variables is fixed.

\end{abstract}




\settopmatter{printfolios=true}
\maketitle

\section{Introduction}
\label{sec-intro}
Reasoning about loops is a foundational task in program analysis and verification.  
Loop invariants  play a crucial and indispensable role; for instance,    
they help establish both safety properties
(as seen in proofs of non-reachability) and liveness properties
(as supporting invariants in termination proofs).
The paper~\cite{beyer2007invariant} goes so far as to call the problem
of automatic invariant generation \emph{the most important task in
program verification}.

The focus of this paper is on the  algorithmic generation of algebraic invariants for
 programs.  
The invariants we
study are given by polynomial equations in the program
variables.
 Not only are these invariants expressive, but they are amenable to a
rich collection of techniques from algebraic geometry.
From a computational perspective, a key question is to determine when the invariant generation problem is decidable. 
For decidability,
the program model must be fairly abstract, as  for instance, 
algebraic invariant generation is already undecidable for polynomial programs~\cite{hrushovski2023strong}.

M\"{u}ller-Olm and Seidl~\cite{Muller-OlmS04} considered the generation of polynomial invariants within the framework of affine programs, raising the question of whether it is possible to compute (a basis of) all polynomial invariants for any given affine program.  This problem can be recast in purely algebraic terms as a question about matrix semigroups: namely, the objective is to compute a representation of the Zariski closure of a finitely generated semigroup of matrices.  If the matrices involved are all invertible then this task is equivalent to that of computing the Zariski closure of a finitely generated matrix \emph{group}.  

The first algorithm to compute the polynomial ideal defining the Zariski closure $\overline{\langle M_1, \ldots , M_k\rangle}$ of the group generated by a set of invertible matrices ${ M_1, \ldots , M_k}$ was introduced in~\cite{DerksenJK05}. 
This algorithm was recently employed in~\cite{HrushovskiOP018, hrushovski2023strong} as a subroutine to  address the above-mentioned question posed by Müller-Olm and Seidl regarding invariant generation for affine programs.
Both of these decision procedures have significant shortcomings in terms of computational complexity; notably, the complexity bound for the group-closure computation in~\cite{DerksenJK05} is not known to be elementary~\cite[Appendix C]{NPSHW2021}. 

A recent advance in computing the group closure was obtained through a linearization technique from~\cite{muller2004note,karr1976affine} in combination with  a novel upper bound  on the degrees of the polynomials defining the  closure~\cite{NPSHW2021}. 
The resulting complexity bounds, although elementary, are of the order of severalfold exponential time for rational matrices; nevertheless there still remains  a significant  gap between the upper and lower complexity bounds.

This paper aims to address and resolve the computational complexity for cyclic matrix groups and semigroups.
Surprisingly, the approaches in~\cite{DerksenJK05,NPSHW2021,HrushovskiOP018, hrushovski2023strong}, and other works such as~\cite{galuppi2021toric,deGraafBook}, fail to achieve reasonable complexity even in the  simplest setting of cyclic matrix groups and semigroups.
This complexity blow-up is a consequence of  using variable elimination  in the computational procedures, which generally has a worst-case exponential-space complexity. 
Alternative approaches based on reductions to the theory of real-closed fields, combined with the degree upper bound obtained in~\cite{NPSHW2021}, will also result in similar complexity bounds. In this case, the complexity blow-up arises from the doubly exponential space required for quantifier elimination in real-closed fields~\cite{davenport1988quantifier}.

In the vocabulary of loop programs, the above simple 
setting translates to that of
branch-free loops with a single linear update, unlike the commonly studied  programs with multiple linear updates~\cite{ kincaid2017reasoning,HrushovskiOP018}. Notwithstanding the radical simplicity of this model, the most natural
verification problems (such as termination and reachability) are
already very difficult in this
setting~\cite{ouaknine2015looptermination}.
We  refer to this  subclass as the \emph{simple linear loops}, and 
 focus on the invariant generation problem for such loops. 

By performing the updates simultaneously,
a simple linear loop
is specified by a single
square matrix~$M$ and a vector \(\bm{\alpha}\) of  initial program values. 
The program state is  given by a vector, and an
iteration of the loop body is summarised by a matrix-vector product. We call a loop with $d$
program variables $d$-dimensional. 
The \emph{orbit} of the loop \(\langle M, \bm{\alpha}\rangle\) is  the reachable set of  program states, defined by
$\mathcal{O} :=
	\{ M^n \boldsymbol{\alpha} : n \in \mathbb{N} \}$.
 Our  main contribution is as follows.

\begin{restatable}{theorem}{theoinvpspace}
\label{theo-inver-pspace}
	Let~$\langle M,\boldsymbol{\alpha}\rangle$ be a rational loop with orbit~$\orbit$.
	A set of polynomials defining 
	$\zski$ is computable in 
 $\PSPACE$.
For all fixed~$d\in \NN$, the computation of $\zski$ for $d$-dimensional loops is in polynomial time. 
\end{restatable}

Loops with linear updates appear in many settings, such as Markov chains and  linear hybrid automata. In \cite{majumdar2020hybrid}, the strongest algebraic invariant is generated for a class of linear hybrid automata, where  the permitted state transitions are linear updates.
Further motivation for complexity results for prototypical classes of loops, as in \cref{theo-inver-pspace}, arises from bottom-up approaches to invariant generation, where one works to summarise larger and larger subprograms in order to analyse expressive program models (cf.~\cite{cyphert2024solvable}).

In general,  the strongest algebraic invariant of a loop with linear updates and equality guards is not computable~\cite{muller2004note}. Similarly, allowing multiple polynomial updates renders invariant generation undecidable~\cite{hrushovski2023strong}, or Skolem-hard \cite{muellner2024strong} with a single polynomial update. Thus linear loops are arguably the richest model where the invariant generation problem admits a complete algorithmic procedure.

\medskip 

\noindent {\bf Applications in Invariant Verification and Loop Synthesis.}
The concept of orbit closure is 
fundamental to numerous subfields of computer science, including geometric complexity
theory, quantum computation,
non-convex optimisation problems,  and graph isomorphism~\cite{forbes2013explicit,blaeser2021tensors,DerksenJK05,burgisser2024completeness,burgisser2011overview}.
Our problem of invariant generation can be interpreted as an implicit orbit-closure problem under the action of a cyclic semigroup~\cite{manssour2024determination}.  

The first application of invariant generation that we consider   is   the \emph{invariant
verification} problem.  In our setting this amounts to checking whether all
reachable states of a given loop satisfy a given collection of
polynomial equations.  We do not assume that the property to be
verified is inductive, and so our analysis involves a non-trivial
examination of reachability in the loop.  The fact that our invariants
are equations plays a key role.  For example, allowing inequalities in
our invariants would render the invariant verification problem more
general than the positivity problem for linear recurrence sequences,
whose decidability status has been open for many decades~\cite{ouaknine2014positivity}.

The second application of our algorithm for invariant generation is   \emph{loop synthesis}.
An example of the kind of scenario we seek to model is as follows.
Imagine that a loop has integer variables $x,y,w$, and $z$.  If, in
each iteration of the loop, both $w$ and $y$ are incremented by $1$,
how should we update variables $x$ and $z$ to 
preserve the invariant $x^2-y^2z^2+z^3=0$?  In other words, we ask to synthesise
variable updates that maintain a given relation among the loop
variables.  In the example at hand, one solution is
$x:=w(y^2-w^2)$ and $z:=y^2-w^2$.  We note that this task is related to
the problem of computing a parametrisation of a variety given by
polynomial equations, which is a classical problem in algebraic
geometry.

We prove that the
invariant verification problem is {\coNP}-complete when the 
polynomials describing the invariant to be verified are given in dense representation.
We also show that the problem
lies in {\PSPACE} when the polynomials are given in sparse representation. In the context of synthesising loops over integers and rationals, building on our work in invariant generation, we consider bit-bounded variants of the synthesis problem and provide  \(\PSPACE\) algorithm for  these cases. Additionally, when the loop dimension is fixed, we establish {\NP} upper bounds and, in some cases, provide matching lower bounds.

A summary of our main results is provided in~\cref{table-results}. 
A comprehensive discussion of these results, along with related previous works, is presented in~\cref{sec-overview}.
%The full version of this paper, including the appendices,  is available at \cite{manssour2024loops}.

% We need layers to draw the block diagram
\pgfdeclarelayer{background}
\pgfdeclarelayer{foreground}
\pgfsetlayers{background,main,foreground}

\def\blockdist{5.5}

\begin{figure*}
\centering
\columnwidth=\linewidth
\begin{tikzpicture}

    %%% LEFT-TOP-HAND SIDE %%%
    %% Invariant Verification
    \node[align=center,text=black] (InvVerif) {\textbf{Invariant Generation}
    \\\textbf{(decision problem)}};
    \node[below=0.5em of InvVerif.south west, align=center,xshift=0.1cm] (dense) {rational loops\\\\ \(\PSPACE\)};
    \node[below=0.5em of InvVerif.south east, align=center,xshift=-0.4cm] (sparse) {rational loops\\fixed dimension\\ \(\PTIME\)};

    %%% LEFT-BOTTOM-HAND SIDE %%%
    %% Invariant Verification
    \node[align=center,text=black,below=2.5cm of InvVerif,xshift=0.05cm] (InvVerif2) {\textbf{Invariant Verification}\\\textbf{(decision problem)}};
    \node[below=0.5em of InvVerif2.south west, align=center,xshift=0.35cm] (dense2) {dense ideal\\ \(\coNP\)-complete};
    \node[below=0.5em of InvVerif2.south east, align=center,xshift=-0.25cm] (sparse2) {sparse ideal\\ \(\PSPACE\)};

    %%% RIGHT-HAND SIDE %%%  
    \path (InvVerif.180)+(8.9cm,-0.1cm) node (LoopSynth) [align=center,text=black] {\textbf{Loop Synthesis over \(R \in \{\Q,\Z\}\)}\\ \textbf{(decision problem)}};

    \node[below=0.2em of LoopSynth.south west, align=center] (weakLoopSynth) {weak\\ as hard as HTP over $R$};
    
    \node[below=0.2em of LoopSynth.south east, align=center] (strongLoopSynth) {strong\\ remains open};
    
    \node[align=center, below=3em of LoopSynth.south] (BitBnd) {Bit-Bounded};
    \node[below=0.2em of BitBnd.south west, align=center,xshift=-0.8cm] (weakBitBnd) {weak\\ \(\PSPACE\) and \(\NP\)-hard};

    \node[below=0.2em of BitBnd.south east, align=center,xshift=0.7cm] (strongBitBnd) {strong\\ \(\PSPACE\)};
    
    \node[align=center, below=3.5em of BitBnd.south] (BitBndfixd) {Bit-Bounded fixed dimension};

    \node[below=0.5em of BitBndfixd.south west, align=center,xshift=+0.5cm] (weakBitBndfixd) {weak\\ $\NP$\\\(\NP\)-hard over $\Z$};

    \node[below=0.5em of BitBndfixd.south east, align=center,xshift=-0.7cm] (strongBitBndfixd) {strong\\ \(\NP\)};

    \begin{pgfonlayer}{background}
                
        \node[draw=black,fit=(BitBnd) (weakBitBnd) (strongBitBnd) (BitBndfixd) (weakBitBndfixd) (strongBitBndfixd) (LoopSynth) (weakLoopSynth) (strongLoopSynth), 
        % fill=blue!10, 
        rounded corners,inner sep=0.4em,minimum width=7.5cm] (BitBndbox) {};
        
        \node[draw=black,fit=(BitBnd) (weakBitBnd) (strongBitBnd) (BitBndfixd) (weakBitBndfixd) (strongBitBndfixd), 
        % fill=blue!12,
        rounded corners,inner sep=0.2em, minimum width=5cm] (BitBndbox) {};
        
        \node[draw=black,fit=(BitBndfixd) (weakBitBndfixd) (strongBitBndfixd), 
        % fill=blue!14, 
        rounded corners,inner sep=0em] (BitBndfixdbox) {};
        
        \node[draw=black,fit=(InvVerif) (dense) (sparse), 
        % fill=blue!10, 
        rounded corners, inner sep=0em,minimum width=5.6cm,line width=0.1cm] {};

        \node[draw=black,fit=(InvVerif2) (dense2) (sparse2), 
        % fill=blue!10, 
        rounded corners, inner sep=0em,minimum width=5.6cm] {};
        
    \end{pgfonlayer}

\end{tikzpicture}
\caption{
Summary of our main results (see~\cref{sec-overview} for a comprehensive overview). 
}
    \label{table-results}
    \vspace{-.3cm}
\end{figure*}

\section{Overview of Main Results}
\label{sec-overview}
In this section, we give a high-level overview of our main results. We will also introduce the basic definitions required to follow  the main techniques exhibited in the algorithms we introduce.
See \cref{sec-app-back}  for extended preliminaries.
%See Appendix A in \cite{manssour2024loops} for extended preliminaries.

We denote by $\ZZ$ and $\QQ$ the set of all integer and rational numbers, respectively.
We write~$\QQ[x]$ for the
ring of univariate polynomials with rational coefficients over~$x$, and 
$\QQ[\boldsymbol{x}]$ for the
ring of multivariate polynomials over variables
$\boldsymbol{x}=(x_1,\ldots, x_d)$.
Given a polynomial $P\in \QQ[\bm{x}]$, 
its total degree is defined as the
maximum total degree of its constituent monomials. Following~\cite{koiran1996hn},  
we define $P(\bm{x})$ as being written in \emph{dense representation} if it is given as an array of coefficients (both zero and nonzero) for all monomials up to its total degree. In contrast, it is in \emph{sparse representation} if it is given as an explicit set of monomials with nonzero coefficients. The size of the polynomial~$P(\bm{x})$, 
in either representation, is the bit length of a reasonable encoding of the polynomial in the required representation, with all numbers written in binary.
For instance, 
the size of $x^{2^n}$ is  $2^n+1$ when written in  dense representation, whereas it is  $n+2$ when written in sparse representation. 
By the above,  the degree of $P(\bm{x})$ in dense representation is 
at most its size, whereas in sparse representation it could be 
exponential in its size.    
Given a set~$S \subseteq \QQ[\boldsymbol{x}]$, 
its description size is  defined as the combination of 
its cardinality, the size of
its polynomials, 
 and the number of variables.

Recall that a complex number  is algebraic 
if it is a root of a nonzero univariate polynomial in~$\QQ[x]$.
Denote by $\alg$  the set of all algebraic numbers. 
The \emph{minimal polynomial} of  \(\alpha\in\QQbar\), denoted by $m_{\alpha}$,   is uniquely defined as the monic polynomial  in ~\(\QQ[x]\) of smallest degree for which \(\alpha\) is a root. 
For algorithmic purposes, we rely on a \emph{symbolic representation} of algebraic numbers\footnote{The symbolic representation of \(\alpha\in\QQbar\)  consists of~$m_{\alpha}$, in sparse representation, combined with  a triple $(a,b,R) \in \mathbb{Q}^3$ such that~$\alpha$ is the unique root of $m_{\alpha}$ that lies  within the $R$-radius circle centered at $(a,b)$ in the complex plane.}, which results in effective arithmetic, see~\cite[Section 4.2.1]{cohen2013course} for more details.

\medskip 

\noindent {\bf Complexity Theory.}
We briefly summarise some relevant notions from complexity theory~\cite{arora2009computational}. 
The  Arthur--Merlin complexity class (\(\AM\)) consists of all decision problems that 
admit a two-round interactive proof in which Arthur tosses some coins and sends the result to Merlin, who responds with a purported proof of membership in the language~\cite{babi1988arthurmerlin}. 
It is well-known that \(\AM\) contains both $\BPP$ (that is, the class of randomised polynomial time with two-sided error) and $\NP$. It is also contained in~$\Pi_2$, the second level of the polynomial hierarchy (\(\PH\)). 

Denote by ${\exists \mathbb{R}}$ the class of problems that are
polynomial-time reducible to the decision problem for the existential
theory of the reals.
Since the latter lies in $\PSPACE$ we have
that ${\exists \mathbb{R}}\subseteq \PSPACE$.
We recall that $\PSPACE$ is closed under $\PSPACE$-oracles, that is,
$\PSPACE^{\PSPACE}=\PSPACE$.

\medskip 

\noindent {\bf Simple Linear Loops.}
A simple linear loop consists of a
single-path loop with linear variable updates.  
For example, consider
the loop in \cref{loop-fibo}.  
It is easy to see
that after $n$ iterations of the \texttt{while} loop, the values of the variables $x$
and $y$ are given by the \((n+1)\)th and $n$th Fibonacci
numbers, respectively.
After \(n\) iterations, the value of variable $z$ is $(-1)^n$.

\begin{figure}[h!]
\vspace{-.2cm}
	\begin{algorithmic}
		\REQUIRE \((x,y,z)\leftarrow (1,0,-1)\)
		\WHILE{\texttt{true}}
		\STATE \begin{equation*}
			\begin{pmatrix} x \\ y \\ z \end{pmatrix} \leftarrow
			\begin{pmatrix}
				1 & 1 & 0\\
				1 & 0 & 0\\
				0 & 0 & -1
			\end{pmatrix}\begin{pmatrix}
				x\\
				y \\
				z
			\end{pmatrix}
		\end{equation*}
		\ENDWHILE
			\end{algorithmic}
			\caption{Fibonacci simple loop program}
			\label{loop-fibo}
			\Description{Fibonacci simple linear loop program}
   \vspace{-.2cm}
		\end{figure}

Recall that  $\textrm{Mat}_d(\QQbar)$ is the set of all $d\times d$ matrices with entries in~$\QQbar$.
A \emph{simple linear loop} $\loopie=\langle M, \boldsymbol{\alpha} \rangle$ is formally defined as a loop program of the form
\begin{equation*}
	\label{def-loop}
	\boldsymbol{x} \leftarrow \boldsymbol{\alpha}; \, 
	\texttt{while}\ \textbf{true} \ \texttt{do}\ \boldsymbol{x}\leftarrow M \boldsymbol{x},
\end{equation*}
where \(\boldsymbol{\alpha} \in \QQbar^d\) is an \emph{initial}  \emph{vector}
and \(M \in \textrm{Mat}_d(\QQbar)\) is an \emph{update matrix}.
The main focus of this paper is on rational and integer loops,
	where
 the constants of $M$ and $\bm{\alpha}$
 lie in~$\QQ$ and in $\ZZ$, respectively.
We define  loops more generally to encompass intermediate steps in our procedures.

The orbit of \(\loopie\) is the reachable set of program states defined by
$\mathcal{O} \coloneqq
	\{ M^n \boldsymbol{\alpha} : n \in \mathbb{N} \} \subseteq \QQbar^d$. The orbit of the program in \cref{loop-fibo} is
\[\left\{ \begin{pmatrix}
	1\\
	0\\
	-1
\end{pmatrix},
\;
\begin{pmatrix}
	1\\
	1\\
	\textover[c]{$1$}{$-1$}
\end{pmatrix},
\;  \begin{pmatrix}
	2\\
	1\\
	-1
\end{pmatrix},
\;  \begin{pmatrix}
	3\\
	2\\
	\textover[c]{$1$}{$-1$}
\end{pmatrix}, 
\; \begin{pmatrix}
	5\\
	3\\
	-1
\end{pmatrix} ,
\; \begin{pmatrix}
	8\\
	5\\
	\textover[c]{$1$}{$-1$}
\end{pmatrix} ,\ldots \right\}.
\]	
A  program with a finite orbit is called trivial, and non-trivial otherwise. 
The  above program  is non-trivial~\cite{kenison2023polynomial}.\footnote{The definition of a non-trivial loop aligns with that of a \emph{wandering point} in the arithmetic dynamics literature~\cite{silverman2007arithmetic,benedetto2019currenttrends}.
}

\medskip 

\noindent {\bf Algebraic Geometry.}  
A polynomial ideal~$I$ is an additive  subgroup of~$\QQ[\boldsymbol{x}]$ 
that is closed under multiplication by  polynomials in~$\QQ[\boldsymbol{x}]$.
Hilbert's Basis theorem states that every polynomial ideal $I \subseteq \QQ[\boldsymbol{x}]$ is
\emph{finitely generated},
equivalently, that every strictly
ascending chain of ideals in $\QQ[\boldsymbol{x}]$ is finite.  
We shall write $I=\langle S \rangle$ if $S\subseteq \QQ[\boldsymbol{x}]$
is a generating set for~$I$.

We view  program orbits  as a subset of the affine space~$\QQbar^d$ for some dimension~$d$.
Following~\cite{HrushovskiOP018,hrushovski2023strong}, we consider  over-approximations of program orbits by  \emph{algebraic sets}.
An algebraic set (or \emph{variety}) is the set of common zeros
of a polynomial ideal~$I$; thus, the  set
\(V(I)\coloneqq \{\boldsymbol{v} \in \QQbar^d :
f(\boldsymbol{v})=0 \text{ for all } f\in I\}\), 
is an algebraic set. By Hilbert's
Basis theorem, every algebraic set can be represented as the set of
common zeros of a finite set
of polynomials. 
An algebraic set $X\subseteq \QQbar^d$ is
\emph{irreducible} if it  cannot be written 
as  $X=X_1\cup X_2$ such that
$X_1$ and $X_2$ are both algebraic  sets and both proper  subsets of~$X$.  In
this paper we use the terms algebraic set and variety
interchangeably whereas certain authors reserve the term variety for
irreducible algebraic sets.

The \emph{strongest algebraic approximation} 
for a program is the smallest algebraic set containing the program orbit. In this context, \emph{smallest} refers to the closure in the Zariski topology 
on~$\QQbar^d$, 
where closed sets are algebraic subsets of~$\QQbar^d$.
For any set~$X \subseteq \QQbar^d$,  its closure in the Zariski topology on~$\QQbar^d$, denoted by $\overline{X}$, is
the smallest algebraic set that contains it.
Subsequently, the strongest  
algebraic approximation of a program is~$\zski$. 

\medskip 

\noindent {\bf Algebraic Invariant.}
Given a loop program $\mathcal{L}=\langle M,\boldsymbol{\alpha}\rangle$, a set $A \subseteq
\QQbar^d$ is an \emph{invariant} if it over-approximates the orbit~$\mathcal{O}$ of $\mathcal{L}$, that is, $\{ M^n \boldsymbol{\alpha} : n \in \mathbb{N} \} \subseteq A$.  
Moreover, if $A$ is an
algebraic set then we call \(A\) an \emph{algebraic invariant}. 
We may refer to the polynomial ideal $I\subseteq \QQ[\boldsymbol{x}]$ such that $V(I)=A$ as 
an invariant ideal for~$\mathcal{L}$.

 We say that a set $A
\subseteq \QQbar^d$ is \textit{inductive} with respect to $\mathcal{L}$ if $\boldsymbol\alpha \in A$
and \(\{ Mv : v\in A\} \subseteq A\).
It is immediate that an inductive
set is an invariant.  It is equally straightforward that the converse
fails in general: not every invariant is inductive. 
However, the loop
$\mathcal{L}$ admits a smallest algebraic invariant, namely the
Zariski closure $\overline{\mathcal{O}}$ of the orbit, and this set
\emph{is} inductive.  Indeed $ M(\zski)
\subseteq \overline{M(\mathcal{O})} \subseteq \zski$ holds by 
Zariski continuity of the self-map $v\mapsto Mv$ on $\QQbar^d$.
A non-trivial algebraic invariant for the  program depicted in
\cref{loop-fibo} is (the zero set of) the polynomial 
\begin{equation}
	\label{eq-invFibbo}
	(y^2+xy-x^2)^2-z^2=0;	
\end{equation}
see, for example~\cite{kauers2008algrel}. 
This set is an invariant by virtue of being inductive.  Indeed, we
have that
the initial point  $(1,0,-1)$
satisfies~\eqref{eq-invFibbo}.
Further, by substitution of the update assignments in \cref{loop-fibo}, we have that $
(x^2+(x+y)x-(x+y)^2)^2-(-z)^2 = (y^2+xy-x^2)^2-z^2 = 0$.
In other words, Equation~\eqref{eq-invFibbo} is stable under the
update matrix of the loop.

\subsection{Algebraic Invariant Generation}
\label{sec-overview-invgen}
Our main contribution {(\cref{theo-inver-pspace})} is  a \(\PSPACE\) algorithm that computes the strongest algebraic invariant of a given simple linear loop. 
This algorithm runs in polynomial time if  the dimension of the loop is fixed.

The geometric properties of the algebraic closure of cyclic subsemigroups of~$\textrm{Mat}_d(\QQ)$ are studied in~\cite{GalSt21toric}. 
Principally, it is shown that each irreducible component of
the  closure of a cyclic semigroup is either an isolated point or isomorphic to a toric variety. Recall that a \emph{toric variety}  is the closed image of a monomial map. 
Toric varieties form an important and rich class of varieties  in algebraic geometry, particularly in view of their combinatorial and algorithmic properties. 
To obtain the result in~\cref{theo-inver-pspace}, we  rely on the observation that the closures of cyclic subsemigroups of matrices and orbit closures of linear loops,    under some  isomorphism of varieties, have the same geometric structure. 
A careful analysis of the results in~\cite{GalSt21toric}, which inspired our research, provides an exponential degree bound for the polynomials defining the closure of matrix semigroups. Such an
exponential degree bound, combined with linearization techniques in~\cite{muller2004note,karr1976affine},
gives an inefficient $\EXPSPACE$ procedure for invariant generation.

In our construction, 
the main obstacle to improve the 
\(\PSPACE\)  bound is 
a basis computation for the lattice of multiplicity relations between the
eigenvalues of the update matrix. The best known bound for this task is through 
a brute-force search in  combination with Masser's bound, see~\cite{DerksenJK05,hrushovski2023strong,galuppi2021toric,NPSHW2021} and \cref{sec:alg}.
%Appendix A in \cite{manssour2024loops}.
%~\cref{sec:alg}. 
If the eigenvalues of the matrix updates are rational numbers, rational multiples of roots of unity, or rational multiples of unnested radicals, this task can be performed more efficiently through identity testing problems for the underlying fields~\cite{BalajiNS022,BalajiPS021}. In those cases, this allows us to place the invariant generation problem in the second level of the polynomial hierarchy~({\PH}).

\medskip

\noindent {\bf Related Work: Invariant Generation.}
Automatically generating invariants for simple linear loop programs has garnered significant attention. This task  remains highly challenging even for simple loop programs with multiple linear updates~\cite{kincaid2017reasoning,HrushovskiOP018,kovacs2008psolvable,humenberger2017hypergeometric}.

There is also a line of work addressing more expressive 
computational models, such as programs with nested loops, conditional branching, probabilistic updates, and unstructured control flow. 
However, these procedures are often limited by the types of polynomial invariants they can generate, rendering them necessarily incomplete. Typically, they produce polynomials only up to a user-defined degree~\cite{amrollahi2024loop, amrollahi2022loop, bayarmagnai2024algebraic, Muller-OlmS04}.
Amongst recent works, a complete method for generating polynomial invariants for extended P-solvable loops is presented in~\cite{humenberger2018psolvable}.
Another recent work~\cite{cyphert2024solvable} generates polynomial invariants for programs with arbitrary control flow, building on an earlier heuristic framework in~\cite{kincaid2017reasoning}. The approach is complete for solvable transition ideals, a generalisation of solvable polynomial maps~\cite{rodriguez2004inv}. However, the Gr\"obner basis calculations in~\cite[Algorithm 3]{cyphert2024solvable} dominate the running time, making the computation exponential~\cite[pg.21]{cyphert2024solvable}.
Recently, invariant generation for moment-computable probabilistic loops was shown to be decidable~~\cite{muellner2024strong}; the proof employs a reduction to invariant generation in our setting (that of simple linear loops).

\subsection{Invariant Verification}
In this subsection, we consider two decision problems concerning invariants for a given loop. The first problem asks whether a set is an inductive invariant;  the second problem asks whether a set is an invariant, 
not necessarily inductive.

Determining  whether an input  ideal
	defines an \emph{inductive set} for a given loop is in principal a  simple task, reducible to \emph{radical membership testing}.
Recall that the radical of a polynomial ideal~$I\subseteq \QQ[\boldsymbol{x}]$, denoted by $\sqrt{I}$, is  obtained by taking all roots of its elements within $\QQ[\boldsymbol{x}]$. By 
Hilbert's Nullstellensatz, if a variety $V$ is such that $V=V(I)$,  the only  polynomials that vanish on $V$ are those in~$\sqrt{I}$.
The \emph{radical membership test} asks to determine whether a polynomial~$P(\boldsymbol{x})\in \QQ[\bm{x}]$ belongs to the radical of a given ideal $I\subseteq \QQ[\boldsymbol{x}]$ (see~\cite[Proposition 8, Chapter 4]{cox2015ideals}).
Radical membership testing reduces to the satisfiability problem of  polynomial equations over~$\QQbar$.
This is a byproduct of  the \emph{Rabinowitsch trick}, see~\cite[Chapter 4.\;§2]{cox2015ideals} for example,  which shows
that $f \in \sqrt{I}$ if and only if 
$V(\ideal{I \cup \{1-yf\}})= \emptyset$  
where $I\subseteq \QQ[\boldsymbol{x}]$ and~$\ideal{I \cup \{1-yf\}} \in \QQ[\boldsymbol{x},y]$.

The satisfiability problem for polynomial equations over~$\QQbar$, 
also known as the \emph{Hilbert's Nullstellensatz} problem (HN for short),
takes as an input a set~$S$ of polynomial equations and asks to determine whether the system  is satisfiable, i.e., whether the
variety of the ideal~$\ideal{S}$
is non-empty.  It is known that
the  HN problem admits an $\AM$ protocol~\cite{Koiran97} under the generalised Riemann hypothesis. This result is independent of the representation of the input polynomials.

By the above, an algorithm to decide whether an input ideal 
$\ideal{S} \subseteq \QQ[\boldsymbol{x}]$  
defines  an inductive set for a given loop~$\langle M,
\boldsymbol{\alpha} \rangle$ must (1)~determine whether~$\boldsymbol{\alpha}$ satisfies all  relations in~$S$, and (2)~test whether $P(M \boldsymbol{x}) \in \sqrt{\ideal{S}}$ for all
polynomials~$P(\boldsymbol{x}) \in S$.
Both of these tests algorithmically reduce to HN.
An alternative approach is to write a query 
in  the existential theory of reals; this leads to  
an unconditional ${\exists \mathbb{R}}$ upper bound.
We note in passing that  for a fixed number of variables such queries can be decided in polynomial time. 

Determining whether an ideal defines an inductive set for a linear loop is less demanding than determining whether an ideal defines an invariant (not necessarily  inductive). 
 We call the  latter the \emph{invariant verification} problem.
A conceptually simple algorithm for the invariant verification
problem is a backward algorithm, used in similar settings in~\cite{benedikt2017polynomial,bayarmagnai2024algebraic,kauers2007equivalence,kauers2008solving}, which examines the ascending chain~$I_0 \subseteq I_1 \subseteq I_2 \subseteq \cdots$ of ideals where 
\begin{equation}
	\label{eq:chain1}
	I_0= \langle S \rangle \quad \text{ and } \quad
	I_i = \langle P(M^j\boldsymbol{x}): P\in S, j\leq i \rangle,
\end{equation}
for all $i\in \NN$. By virtue of $\QQ[\boldsymbol{x}]$ being Noetherian,
the ascending sequence 
of nested ideals in~\eqref{eq:chain1} stabilises. In our setting, by properties of linear transformations and  the well-known fingerprinting procedure for solving Algebraic Circuit Identity Testing (ACIT)~\cite{allender2009complexity}, we obtain the following.

\begin{restatable}{proposition}{lemweakverifycoNP}
\label{lem-weakverify-coNP}
	Given 
	$S \subseteq \QQ[\boldsymbol{x}]$  in dense representation and a  loop  with orbit~$\orbit$, 
	verifying
	whether $\mathcal{O} \subseteq V(S)$ is \(\coNP\)-complete.
\end{restatable} 

The above na\"ive algorithm results in an inefficient \(\EXPSPACE\) upper bound  when the polynomials in~$S$ have sparse representation. One of our contributions shows that invariant verification is in \(\PSPACE\), even if the input  polynomials are given in sparse representation. 
Our procedure also verifies whether  $\zski =V(S)$, a task that 
cannot be performed through the study of the ideal chain in~\eqref{eq:chain1}.

\begin{restatable}{proposition}{lemweakverify}
	\label{lem-weakverify}	 
	Given 
	$S \subseteq \QQ[\boldsymbol{x}]$ and   a loop with the orbit~$\orbit$, verifying whether
	$\mathcal{O} \subseteq V(S)$ holds is in ${\PSPACE}$, and is 
	${\coNP}$-hard. Similarly, the test $\zski =V(S)$ can be performed in ${\PSPACE}$.  
	For all fixed $d\in \mathbb{N}$, both problems for $d$-dimensional loops are decidable in polynomial time. 
\end{restatable}

\medskip 

\noindent {\bf Related Work: Invariant Verification.}
 Invariant verification in the setting of matrix
semigroups (a non-deterministic model of loop programs) was shown to
be decidable by Dr\"{a}ger~\cite{drager2016invariance}.  While the
program model considered in~\cite{drager2016invariance} is more
expressive than linear loops (which essentially correspond to matrix
semigroups with a single generator), the class of invariants is restricted---the invariants therein are given by conjunctions of linear
equations.

\subsection{Weak and Strong Loop Synthesis}
Another application for the computation of strongest invariants, is the complementary view of 
 synthesising linear loops from a given   invariant. 
Most generally, the loop synthesis problem asks, given a polynomial ideal with 
 generating set $S\subseteq \QQ[\boldsymbol{x}]$, whether there exists a  linear loop with (an infinite) orbit~$\orbit$ such that 
$\zski \subseteq V(S)$. 
This problem and its  variants have recently received  much attention~\cite{humenberger2020algebra,humenberger2022algebra,hitarth2024quadratic}; see the discussion in related works below.

A recently defined variant of the loop synthesis problem sharpens the inclusion of $\zski$ in $V(S)$ by requiring equality of these two varieties.
We call this variant  the \emph{strong synthesis} problem, whereas we may refer to the loop synthesis problem as the \emph{weak synthesis problem}. As an instance, 
recall the  ideal~$I\coloneqq\langle (y^2+xy-x^2)^2-z^2 \rangle$ generated by the polynomial in~\eqref{eq-invFibbo}. 
Given~$I$ as the input, 
the loop program  in \cref{loop-fibo} is a witness for weak synthesis  but not for strong synthesis. 
Indeed, the strongest polynomial invariant for the loop in \cref{loop-fibo} is the ideal generated by $J=\langle y^2+xy-x^2-z,z^2-1 \rangle$ for which $I$ is a strict subset (meaning that $V(J)\subset V(I)$).

We study both weak and strong synthesis for  loops over $\ZZ$ and $\QQ$. 
An informal discussion in~\cite[Remark~2.8]{hitarth2024quadratic}
connects loop synthesis and  Diophantine equations. 
The immediate observation in \cref{claim:undec} formalises this connection.
The proof of \cref{claim:undec} is given in \cref{sec:omitted}.

\begin{restatable}{lemma}{claimundec}
	\label{claim:undec}
	The weak synthesis problem 
	over $\{\ZZ,\QQ\}$ is as hard as Hilbert’s Tenth problem (HTP) over $\{\ZZ,\QQ\}$, even in fixed dimension.
\end{restatable}

HTP over~$\ZZ$ was shown undecidable by Matiyasevich in 1970~\cite{matiyasevich1970h10}, even for polynomial equations with fixed number of variables (as small as  11 variables)~\cite{poonen2008}.
The decidability of HTP over~$\QQ$ is  a long-standing open problem in  the theory of Diophantine Equations~\cite{shlapentokh2011h10} and in arithmetic geometry~\cite{poonen2008}. 
\cref{claim:undec} implies undecidability of the weak synthesis of loops over~$\ZZ$.  
In~\cite{davis1972numsol}, it is shown that 
the problem of asking whether a variety has infinitely many integer points is as hard as  HTP over~$\ZZ$. 
The reduction in ~\cite{davis1972numsol} can be extended naturally  to show that the problem of asking whether a variety has infinitely many rational points is also as hard as HTP over~$\QQ$. The strong synthesis of  
loops over $\ZZ$ and $\QQ$ has the flavour of this latter problem, and we leave the decidability open.

\medskip 

\noindent {\bf Bit-bounded Synthesis.} 
The HTP-hardness of the (weak)  synthesis problem  motivates the following notion of \emph{bit-bounded synthesis.}
The bitsize of a rational number $\frac{a}{b}$, with
$a$ and $b$ co-prime integers, is  $\log(|a|)+\log(|b|)$.
Given a linear loop~$\loopie=\langle M, \boldsymbol{\alpha} \rangle$,
we say that it 
is \emph{$B$-bounded rational}, if  all entries of~$M$ and $\alpha$
are rational numbers with bitsize at most~$B$.\footnote{\(B\)-boundedness does not place any restrictions on higher powers of \(M\) nor the orbit of the associated loop \(\loopie\).}

The \emph{strong bit-bounded synthesis problem} asks, given a set~$S \subseteq \QQ[\boldsymbol{x}]$ of polynomials and a bound~$B\in \NN$, 
whether there exists some $B$-bounded pair~$\langle M, \boldsymbol{\alpha} \rangle$
with orbit~$\orbit$ such that 
$\zski= V(S)$.
Similarly, the \emph{weak bit-bounded synthesis problem} asks
whether~$\zski\subseteq  V(S)$ holds.

We place the bit-bounded variants of loop synthesis, for both strong and weak cases, in $\PSPACE$ through the invariant verification problem.
We also establish complexity lower bounds for these problems through
 reductions from $\mathsf{3SAT}$
 and $\mathsf{Unique\;SAT}$.

\begin{restatable}{proposition}{thbitbounded}
	\label{th:bitbounded}
	The strong and weak bit-bounded synthesis problems over~$\{\QQ,\ZZ\}$ lie in $\PSPACE$. The weak variant  is \(\NP\)-hard; and the strong variant is \(\NP\)-hard under randomised reductions.
\end{restatable}

When synthesising loops with a fixed dimension~$d\in \NN$,
we obtain an $\NP$ upper bound.
For weak synthesis over~$\ZZ$ we show that the problem is complete for the class~$\NP$.

\begin{restatable}{proposition}{thbitboundedfix}
	\label{th:bitboundedfix}
	For all fixed $d\in \NN$, the strong and weak bit-bounded synthesis problems in dimension $d$ over~$\{\QQ,\ZZ\}$ are in $\NP$. Moreover, weak bit-bounded synthesis over~$\ZZ$ is $\NP$-complete.
\end{restatable}  

\medskip 

\noindent {\bf Related Work: Loop Synthesis.}
Program synthesis, conceived as the problem of generating constraints
that relate unknowns and enforce correctness requirements, has
received significant
attention~\cite{alur2018search,gulwani2017program}.  Recent works
focusing on polynomial invariants \cite{hitarth2024quadratic,humenberger2022algebra,humenberger2020algebra, kenison2023polynomial}
have leveraged algebraic techniques to recast the problem of loop
synthesis as that of solving an algebraic system of recurrence
sequences.  The template-based procedure in the work by Humenberger et
al.~\cite{humenberger2022algebra,humenberger2020algebra} finds the
solution to this system of recurrences by solving a polynomial
constraint problem; however, we note their solutions result in loops
defined over \(\alg\).  A Diophantine approach to loop synthesis is employed in
\cite{hitarth2024quadratic}; therein those authors synthesise loops when
given a single quadratic equation as an input.  Another recent work~\cite{kenison2023polynomial} 
gave a procedure for synthesising loops for an input binomial ideal.

Matrix completion is the task of completing a partially defined matrix according to a given specification~\cite{johnson1990matrix}.
Perhaps the most notable variant is the \emph{Netflix Problem}~\cite{Candes2010} (an application in collaborative filtering~\cite{goldberg1992collaborative}).   
In this application the
goal is to complete a matrix of movie recommendations so as to
minimise the rank (or some proxy thereof such as the nuclear norm).
This is reminiscent of template-based approaches towards program
synthesis~\cite{srivastava2013template}.
In the language of loop synthesis,   the partial matrix represents an
incomplete program fragment and the task is to complete the program so
as to guarantee certain desired polynomial invariants.

\section{Algebraic Closure of Linear Loops}
\label{sec-algLoops}
Our goal in this section is to construct the strongest algebraic invariant of an input loop. 
Fix a loop~$\loopie$ with update matrix~\(M\in\QQ^{d\times d}\) and initial vector~\(\bm{\alpha}\in\QQ^d\).
We define the size of $\loopie$ as $d + \ell$, where  $\ell$ is the bitsize of the entries of $M$ and $\boldsymbol{\alpha}$.

As in~\cite{DerksenJK05,hrushovski2023strong,galuppi2021toric,NPSHW2021}, our first step is to compute the Jordan normal form 
of the update matrix~$M$ via a Jordan decomposition.
A Jordan decomposition of~$M$ comprises  a Jordan  matrix $J$, and a  change-of-basis matrix~$P$ which satisfy $M = P J P^{-1}$.
Recall that the matrix $J$ is such that  the only nonzero entries of $J$ are on its diagonal and its superdiagonal, and that  
the matrix~$P$ in the Jordan decomposition is not unique.
In \Cref{sec:JBD}, roughly speaking, we show that~$P$ can be chosen 
such that $P^{-1} \bm{\alpha}$ is a binary vector.
This  simplifies the computation of the invariant equations, derived
in~\Cref{sec:computestrong}, and our proposed algorithms.  
We analyse the computational complexity of the 
algorithm in~\Cref{sec:complexity}.
Worked examples demonstrating the algorithm are given in~\cref{sec-worked-examples}.

\subsection{A Convenient Jordan Block Decomposition}
\label{sec:JBD}

The \emph{Jordan normal form} \(J\in\alg^{d\times d}\) of \(M\) is given by  the direct sum of several so-called  Jordan blocks 
$J_1, \dots, J_s$. 
Each block $J_i$ is a square matrix of dimension~$d_i$ as follows:
\[
J_i \coloneqq \begin{pmatrix}
	\lambda_i & 1 & & \\
	& \ddots & \ddots & &\\
	&  & \lambda_i & 1\\
	& & & \lambda_i \\
\end{pmatrix}.
\]
The multiset  $\{\lambda_1, \dots, \lambda_s\}$ of algebraic numbers is  equal to  the multiset of eigenvalues of~$M$. Indeed, for each   eigenvalue  $\lambda_i$ the sum of the sizes of all Jordan blocks corresponding to~$\lambda_i$ is its algebraic multiplicity.
Recall  that 
\(J\) is unique up to the ordering of its Jordan blocks.

A \emph{Jordan decomposition} of \(M\), denoted by $(P,J)$, consists of the associated Jordan normal form \(J\) and an invertible change-of-basis matrix \(P\in\alg^{d\times d}\) for which \(M=PJP^{-1}\).
According to the $s$ blocks of~$(P,J)$, we can partition  $d$-dimensional vectors $\bm{v}$ into $s$ block vectors $\bm{v}_i$ such that 
\begin{equation}
	\label{eq-partblock}
	\bm{v}\coloneqq(\bm{v}_1,\ldots, \bm{v}_s)
\end{equation} 
where  each  $\bm{v}_i$  is a vector of size~$d_i$.
The \emph{fingerprint} of~$\bm{\alpha}$ with respect to~$(P,J)$ is
defined as a binary vector~$\bm{\beta}\in \{0,1\}^d$
where, for each block $\bm{\beta}_i=(\beta_{i,1},\ldots, \beta_{i,d_i})$ with $i\in \{1,\ldots,s\}$,
\begin{itemize}
	\item the entry~$\beta_{i,j}$ is $1$ if and only if
	$j$ is the largest index  such that $\alpha_{i,j}$ is nonzero. 
\end{itemize}
We define a Jordan decomposition~$(P,J)$ as  \emph{convenient} for~$\langle M, \bm{\alpha}\rangle$ if $P^{-1} \bm{\alpha}$ 
is the fingerprint of~$\bm{\alpha}$ with respect to the decomposition.

\begin{restatable}{lemma}{oneinblock}
	\label[lemma]{oneinblock}
	For a loop~$\langle M, \bm{\alpha} \rangle$ and any Jordan decomposition~$(P,J)$ of  $M$,
	there exists a matrix $U \in \GL_d(\alg)$ such that 
	$(PU^{-1},J)$  is a convenient Jordan 
	decomposition for~$\langle M, \bm{\alpha} \rangle$. The computation of~$U$ is in polynomial time in the size of the loop.	
\end{restatable}

\begin{proof}
   
    Let \((P,J)\) be a Jordan decomposition of \(M\), 
wherein all algebraic numbers are represented \emph{symbolically}. The computation of \((P,J)\) and $P^{-1}$ can be performed in polynomial time~\cite{cai1994jnf,cai2000complexityabc}.

    Write 
    $J = J_1 \oplus \cdots \oplus J_s$ as a direct sum of Jordan blocks.  Consider the corresponding block decomposition of the vector $P^{-1} \boldsymbol{\alpha} = \boldsymbol{\alpha}_1 \oplus \cdots \oplus \boldsymbol{\alpha}_s$. We  construct a matrix $U = U_1 \oplus \cdots \oplus U_s$ such that $UP^{-1}\boldsymbol{\alpha} \in \{0,1\}^d$ and $U$ commutes with $J$. To this end, suppose that  $\boldsymbol{\alpha}_i $ is given by $\boldsymbol{\alpha}_i = (\alpha_{1,i}, \ldots, \alpha_{r,i}, 0\ldots, 0)^\top$ such that $\alpha_{r,i} \neq 0$. We define
\[ 
U_i \coloneqq \begin{pmatrix}
    b_r & \cdots  & b_1     &     \cdots     &  0 \\
        & \ddots  &         &  \ddots   &  \vdots \\
        &         & \ddots  &            & b_1 \\
        &         &          &  \ddots   & \vdots  \\
        &         &          &           &b_r
\end{pmatrix}
\]
such that each \(U_i\) is upper triangular and  Toeplitz (meaning that 
along each diagonal the entries are constant).
The entries $b_1,\ldots,b_r$ are chosen such that 
    $U_i \boldsymbol{\alpha}_i = \bm{e}_r$, where $\bm{e}_r$ is the standard unit vector with $1$ in the $r$th position. 
    Specifically,  $b_r = 1/{\alpha_{r,i}}$, and the other $b_j$'s are defined one-by-one for $j=r-1,\ldots,1$ by back substitution. 
    
    By construction we have $UP^{-1}\boldsymbol{\alpha} \in \{0,1\}^d$.
    By a well-known property of upper triangular Toeplitz matrices, each $U_i$ commutes with $J_i$, implying  that~$J U = U J$.  Define $Q \coloneqq PU^{-1}$, and observe that $M =P J P^{-1} = Q J Q^{-1}$.  Moreover,  $Q^{-1}\boldsymbol{\alpha}\in \{ 0, 1\}^d$ holds and  $Q^{-1}\boldsymbol{\alpha}$ is the fingerprint of $\bm{\alpha}$. Hence, the Jordan decomposition $(PU^{-1}, J)$ is convenient for $\langle M, \boldsymbol{\alpha} \rangle.$ 
    \end{proof}

\subsection{Computing the Strongest Algebraic Invariant} 
\label{sec:computestrong}

Let $(P,J)$ be an arbitrary Jordan decomposition of~$M$ such that all Jordan blocks with zero eigenvalue appear first.
Consider the orbit \(\mathcal{O}\coloneqq\{ M^n \boldsymbol{\alpha} : n \in \mathbb{N} \}\) of the loop. We are interested in computing 
the polynomials that define  the algebraic set 
\(
\zski=\overline{\{M^n \boldsymbol{\alpha} : n \in \mathbb{N} \}}.
\)
Since
\(M^n\boldsymbol{\alpha} = PJ^nP^{-1}\boldsymbol{\alpha}\)
holds for all~$n\in \NN$,
we  compute the polynomials that define \(\zski\) in two steps:
\begin{itemize}
	\item first, we can compute the polynomials that define the algebraic set
	\(\overline{\{ J^n \; (P^{-1}\boldsymbol{\alpha}): n \in \mathbb{N} \}}\) and
	\item second, we apply the invertible linear transformation \(P\).
\end{itemize}

For the  decomposition~$(P,J)$ of~$M$,
let $N$ be the direct sum of all Jordan blocks associated with   eigenvalue~$0$.
Let $d_0$ be the dimension of~$N$; that is, the algebraic multiplicity of~$0$. Write
$J_1, \dots, J_s$ for all Jordan blocks  with associated nonzero eigenvalues
$\lambda_1, \dots, \lambda_s$, respectively.  
Denote by $d_i$ the dimension of the Jordan blocks~$J_i$. 
Thus
	\(J= N \oplus J_1 \oplus \cdots \oplus J_s\).

The above-mentioned matrix~$N$  is \emph{nilpotent}, i.e., there is an integer $m \leq d_0$ such that $N^m = 0$. 
Thus  the set \(\{ N^n : n \in \mathbb{N} \}\) is finite, which in turn implies that the  algebraic set
$\overline{\{ J^n \; (P^{-1}\boldsymbol{\alpha}): n \in \mathbb{N} \}}$ contains at most  $m$ isolated points: 
\begin{equation}
	\label{def-isp}
	\{P^{-1}\boldsymbol{\alpha} , J P^{-1}\boldsymbol{\alpha} , \ldots, J^{m -1} P^{-1}\boldsymbol{\alpha}\}.
\end{equation} 
Write $n_0$ for the number of distinct isolated points in the above set.
Define 
\(\bm{\gamma}\coloneqq J^{n_0} \; (P^{-1}\boldsymbol{\alpha})\). 
We obtain $\Tilde{J}$ from \(J = N \oplus J_1 \oplus \cdots \oplus J_s\) by replacing $N$ with the  zero matrix (of   size
$d_0\times d_0$).
By~\eqref{def-isp} and the definitions of~$\bm{\gamma}$ and $\Tilde{J}$,  
the  algebraic set
$\overline{\{ J^n (P^{-1}\boldsymbol{\alpha}): n \in \mathbb{N} \}}$ decomposes into 
\[\{J^i P^{-1}\boldsymbol{\alpha} \, :\, 0\leq i< n_0\}\cup  \overline{\{ \Tilde{J}^n \, \bm{\gamma}: n \in \mathbb{N} \}}.
\]
This allows us  to   first focus on the  Zariski closure of the loop with the invertible transition   matrix~$J_1 \oplus \cdots \oplus J_s$, and then recover~$\zski$. 

Towards this goal, we apply Lemma~\ref{oneinblock} to $\langle \Tilde{J}, \boldsymbol{\gamma} \rangle$ in order to compute the matrix~$U$ 
for which $(PU^{-1},\Tilde{J})$ is a convenient Jordan decomposition
for~$\langle \Tilde{J}, \bm{\gamma} \rangle$, which by  construction
respects the ordering of the Jordan blocks in~$\Tilde{J}$.
Define
\(\bm{\beta}\coloneqq U\bm{\gamma}\).  
Consider the partition of $\bm{\beta}$ according to the decomposition~$(P,J)$, defined in \eqref{eq-partblock}. Denote by 
$\bm{\beta}_i$ the block of $\bm{\beta}$   
that corresponds to the Jordan block~$J_i$.
Thanks to  the convenient Jordan decomposition, each block~$\bm{\beta}_i$ is either a zero vector, or a standard unit vector.
For block~$\bm{\beta}_i$, we refer to $k_i$ as the index of the nonzero entry $\beta_{i, k_i}$ of~$\boldsymbol{\beta}_i$.  
(We assume $k_i = 0$ in the case that $\beta_i$ is a zero vector.)

Consider $ i\in \{ 1, \ldots, s\}$ such that the block~$\bm{\beta}_i$ is a standard unit vector. 
By analysing $J_i^n \bm{\beta}_i$, we find that
\begin{equation}\label{eq:block}
	J_i^n \bm{\beta}_i=   \begin{pmatrix}
		\binom{n}{k_i -1} \lambda_i^{n-k_i+1}\\
		\binom{n}{k_i -2} \lambda_i^{n-k_i+2}\\
		\vdots\\
		n \lambda_i^{n-1} \\
		\lambda_i^n\\
		0\\
		\vdots\\
		0
	\end{pmatrix}
\end{equation}
holds for  all $n \in \N$.

At the next step we introduce a linear transformation $R \coloneqq I_{d_0} \oplus R_1 \oplus  \cdots \oplus R_s$
to simplify~\eqref{eq:block} such that 
\begin{equation}\label{eq:only-powers-n}
	R_iJ_i^n \bm{\beta}_i=   \begin{pmatrix}
		n^{k_i - 1} \lambda_i^n\\
		n^{k_i - 2} \lambda_i^n\\
		\vdots\\
		n \lambda_i^n \\
		\lambda_i^n\\
		0\\
		\vdots\\
		0
	\end{pmatrix}.
\end{equation}
Recall the combinatorial identity
\begin{equation}\label{eq:stirling2nd}
	n^k = \sum_{i=1}^{k} i!  \genfrac{\{}{\}}{0pt}{}{k}{i}  \cdot \genfrac{(}{)}{0pt}{}{n}{i} = \sum_{i=1}^{k} c_{k,i} \genfrac{(}{)}{0pt}{}{n}{i},
\end{equation}
where $ \genfrac{\{}{\}}{0pt}{}{k}{i} = \frac{c_{k,i}}{i!}$ is the Stirling number of the second kind.
Since the Stirling numbers are defined recursively by the relation \(c_{k+1,i} = i (c_{k,i} + c_{k,i-1})\),
we can compute the coefficients in~$R_i$, starting with $c_{1,1} = \dots = c_{k_i-1,1} = 1$, in polynomial time.
The block matrix \(R_i\) is defined below.

	\begin{equation} \label{eq:Ri}
		\renewcommand{\arraystretch}{1}
		R_i \coloneqq 
		\left( \begin{array}{@{}cccccc|c@{}}
			
			c_{k_i-1,k_i-1}\lambda_i^{k_i-1} & c_{k_i-1,k_i-2}\lambda_i^{k_i-2} & &\dots & c_{k_i-1,1} \lambda_i & 0 &\\
			& c_{k_i-2,k_i-2}\lambda_i^{k_i-2} & & & c_{k_i-2,1} \lambda_i & \vdots & \\
			& & \vphantom{\ddots} & & \vdots & & \\
			& & \ddots & & & & \\
			&  &  &c_{2,2}\lambda_i^2 & c_{2,1}\lambda_i & & \\
			& & & & c_{1,1}\lambda_i & 0 & \\
			&  & & &  & 1 & \\
			\hline
			& & & & & & I_{d_i-k_i}
		\end{array} \right).
	\end{equation}

One can  see that \eqref{eq:only-powers-n} is a direct consequence of \eqref{eq:block}, \eqref{eq:stirling2nd} and \eqref{eq:Ri}.
We note that, by construction, $R$ is an invertible matrix.

Let $\bm{x}=(x_1,\ldots,x_d)$ where $d$ is the dimension of the matrix~$M$.  
Consider the partition of $\bm{x} = (\bm{x}_0, \ldots, \bm{x}_s)$ according to the decomposition~$(P,J)$, defined in \eqref{eq-partblock}. Let
\(\bm{x}_{i}=(x_{i,1},\ldots,x_{i,d_i})\)
be the block vector corresponding  to the Jordan block~$J_i$ with $i\in \{1,\ldots,s\}$
and  
\(\bm{x}_{0}=(x_{0,1},\ldots,x_{0,d_0})\) the concatenation of the block vectors  corresponding to Jordan blocks with zero eigenvalues.

In the following, we construct a set of polynomials in $\QQ[\bm{x}]$ that define the variety	
\begin{equation}
	\label{eq-r-closure}
	R \cdot  \overline{\{ \Tilde{J}^n \, \bm{\beta}: n \in \mathbb{N} \}},
\end{equation}
which in turn helps us define the set of polynomials for  the algebraic  closure of the orbit of~$\langle M, \bm{\alpha}\rangle$.

Recall that the only nonzero entry of the block~$\boldsymbol{\beta}_i$ is indexed by $k_i$ if $\boldsymbol{\beta}_i$ is a standard unit vector, and 
$k_i = 0$ otherwise.
Define the set~$S_1$ of polynomials in 
$\QQ[x_{1,k_1}, \ldots, x_{s,k_s}]$ such that 
\begin{equation}
	\label{def-s1}	
	V(S_1) = \overline{\{(\lambda_1^n, \ldots, \lambda_i^n, \ldots , \lambda_s^n)\, :\, n \in \N; \, \bm{\beta}_i\neq 0\}}.
\end{equation}
The task of computing the defining polynomials is well-understood, see for example~\cite{DerksenJK05,hrushovski2023strong,galuppi2021toric,NPSHW2021} and
% Appendix A in \cite{manssour2024loops} for details.
\cref{sec:alg} for details. 
The underlying idea is to compute a basis for the lattice of  multiplicity relations between the eigenvalues.

For each Jordan block~$J_i$ with $i\in \{1,\ldots,s\}$,
define the set~$S_{2,i}$ of polynomials in 
$\QQ[\bm{x}_{i}]$ as follows: 
 \begin{equation*}
	S_{2,i} \coloneqq 
	\begin{cases}
		\{ x_{i, k_{i}-1}^j \,  - x_{i, k_{i}-j} x_{i, k_i}^{j-1} \, : \, 2 \leq j \leq k_i -1 \} & \text{if $k_i\geq 3$,}\\
		\emptyset & \text{otherwise}.
	\end{cases}
\end{equation*}
This set captures the relations between the nonzero entries of the $i$th block, as in~\eqref{eq:only-powers-n}.

For each pair of distinct Jordan blocks~$J_i,J_j$, with $i,j\in \{1,\ldots,s\}$, 
we define the set~$S_{3,i,j}$ of polynomials in 
$\QQ[x_{i, k_i-1}, x_{i, k_i}, x_{j, k_{j}-1 }, x_{j, k_j}]$ as follows: 
\begin{equation*}
	S_{3,i,j} \coloneqq 
	\begin{cases}
		\{ x_{i, k_i - 1 } x_{j, k_j} - x_{j, k_{j}-1 } x_{i, k_i}\} & \text{if $k_i\geq 2$ and $k_j\geq 2$,}\\
		\emptyset & \text{otherwise}.
	\end{cases}
\end{equation*}
This set
defines the relations  between the two blocks~$J_i$
and $J_j$. Such a relation only exists when  in both blocks~$\boldsymbol{\beta}_i$ and 
$\boldsymbol{\beta}_j$, the indices $k_i$ and $k_j$ corresponding to the nonzero entries are both greater than or equal $2$.

Define the set~$S_{4}$ of polynomials in 
$\QQ[\bm{x}]$ such that  
 \begin{equation*}
	S_{4}\coloneqq \{ x_{0,j} \, : \, 1\leq j \leq d_0\}
	\cup \{ x_{i,j} \, : \, 1\leq i\leq s, k_i< j \leq d_i \}  
\end{equation*}
captures the zero entries.
Define $I_{R}\coloneqq\ideal{S_1, \, S_{2,i}, \, S_{3,i,j},S_4 \,:\, i,j\in \{1,\ldots,s\}}$. We are now in the position to prove the following: 
\begin{claim}\label{claim-IR}
	$V(I_R)=R \cdot  \overline{\{ \Tilde{J}^n \, \bm{\beta}: n \in \mathbb{N} \}}$.
\end{claim}

\begin{proof}
	
	Let $\bm{v}_n$ be the vector obtained from $\Tilde{J}^n \, \bm{\gamma}$ 
	under the transformation~$R$, shown explicitly in~\eqref{eq:only-powers-n}. 
	We consider the block decomposition of $\bm{v}_n$ according to~$(P,J)$. 
	
	We did not precisely define the set of polynomials in~$S_1$ as 
	we borrow the technology developed in~\cite[Lemma 6]{DerksenJK05}
	to construct $S_1$ defining  the  variety in~\eqref{def-s1}. 
	The polynomials in $S_1$ are defined over the variables~$x_{1,k_1}, \ldots, x_{s,k_s}$. 
	This choice is justified by the fact that  in the block corresponding to the Jordan block~$J_i$
	of~$\bm{v}_n$ the $k_i$th entry is $\lambda_i^n$. 
	
	For each Jordan block~$J_i$, the set $S_{2,i}$
	reflects the relations between corresponding entries in~$\bm{v}_n$.
	Inspecting~\eqref{eq:only-powers-n} again for the  entries of~$\bm{v}_n$,  the following identity  
	\begin{align*}
		\left(n\lambda_i^n\right)^j \,  - n^j \lambda_i^n \left(\lambda_i^n\right)^{j-1} \,=\,0  	
	\end{align*}
	is realised by $x_{i, k_{i}-1}^j - x_{i, k_{i}-j} \cdot x_{i, k_i}^{j-1}$, with $j \in \{2,\ldots,k_i-1\}$.
	
	Furthermore, for each pair of distinct Jordan blocks~$J_i,J_j$, with $i,j\in \{1,\ldots,s\}$, 
	the set~$S_{3,i,j}$ encompasses the relations between the two blocks.
	As mentioned above, there is no relation between ~$J_i$ and $J_j$ if 
	either $k_i<2$ or $k_j<2$. Otherwise, the identity
	\begin{align*}
		(n \lambda_i^n) \, (\lambda_j^n) -  (n \lambda_j^n) \, (\lambda_i^n)\,=\,0  	
	\end{align*} 
	is realised by $x_{i, k_i - 1 } x_{j, k_j} - x_{j, k_{j}-1 } x_{i, k_i}$.
	The set~$S_4$ reflects the entries in $\bm{v}_n$ that are identically zero. 
\end{proof}

In \cite{galuppi2021toric} it was shown that each irreducible component of the variety defined by a cyclic matrix semigroup is isomorphic to the  Cartesian product of a toric variety and a  normal rational  curve (excluding  a number of isolated points). Our contribution in this regard is to show
 an explicit construction of that isomorphism (in the case of loops)  can be computed in  polynomial time using the matrices $R$ and $U$.
In \eqref{eq:only-powers-n} we read-off  the effect of the  Cartesian product of a toric variety and a  normal rational  curve.

It remains to define a generating set of polynomials 
for~\(\zski\) from the  ideal $I_R$. 
From the steps detailed in the construction, we observe that  
 \begin{equation}\label{orbit-closure-is}
	\overline{\{M^n \boldsymbol{\alpha} : n \in \mathbb{N} \}} = P   U^{-1} R^{-1} R \cdot  \overline{\{ \Tilde{J}^n \, \bm{\beta}: n \in \mathbb{N} \}} 	\cup \{M^i \boldsymbol{\alpha} : 0\leq i< n_0\}.
\end{equation}
By this equality, to obtain the 
polynomials defining $\zski$, 
\begin{itemize}
	\item we first compute ideals $\mathcal{I}_0, \dots, \mathcal{I}_{n_0-1}$, each defining one of the $n_0 \leq d$ isolated points in $\{M^i \boldsymbol{\alpha} \, :\, 0\leq i < n_0\}$,
	\item next we obtain an ideal $\mathcal{J}$ by applying   the transformation~$P U^{-1} R^{-1}$ to the 
	vector of variables~$\bm{x}$ for each polynomial in the set
	\begin{equation}
		\label{eq-SS}
		S\coloneqq (S_1\cup S_4) \cup  \bigcup_{\substack{i,j\in \{1,\ldots,s\} \\i\neq j}} (S_{2,i}\cup  S_{3,i,j})
	\end{equation} 
\end{itemize}
defining $V(I_R)$.
This  completes the construction of the strongest algebraic 
invariant of the loop~$\langle M, \bm{\alpha}\rangle$
as it is simply~$V(\mathcal{I}_0 \cap \cdots \cap \mathcal{I}_{n_0-1} \cap \mathcal{J})$.
For complexity purposes, in particular to avoid an exponential blow-up, 
we choose to output the generators of the~$\mathcal{I}_i$ and $\mathcal{J}$ separately. (This choice is crucial for the obtained $\PSPACE$ bound in~\Cref{lem-weakverify}.)

The obtained polynomials are not in $\QQ[\bm{x}]$ but in $k[\bm{x}]$ where $k=\QQ(\lambda_1,\ldots,\lambda_s)$ is the number field obtained by adjoining the eigenvalues to~$\QQ$. 
Furthermore, the eigenvalues are given in symbolic presentation.
As an extra step (not necessary to obtain the complexity bound) we convert the coefficient of polynomials to integers.
We introduce the variables~$y_i$, one for  each  eigenvalue~$\lambda_i$.	
Recall that  each $\lambda_i$ appears with all its Galois conjugates, and its degree is bounded from above by $d$.
For brevity, given a Galois automorphism~$\sigma$
of the field $\Q(\lambda_i)$
we denote by 
$\sigma(y_i)$ the variable among $y_1,\ldots,y_s$
corresponding to~$\sigma(\lambda_i)$. 
For each~$\lambda_i$, write   $m_{\lambda_i}=\sum_{k=0}^n a_k x^k$ for the minimal polynomial. 
By basic algebra, we know that the set of equations
\begin{equation}
	\label{eq-eign}
	\begin{aligned}
		a_{n-1}&=- \sum_{\sigma \in {\text{Gal}}(\QQ(\lambda_i)/\QQ)}
		\sigma(y_i)\\
		a_{n-2}&=  \sum _{\substack{\sigma,\sigma' \in {\text{Gal}}(\QQ(\lambda_i)/\QQ)}} \sigma(y_i)\, \sigma'(y_i)\\
		&{}\ \,\vdots \\
		a_{0}&=(-1)^{n}  \prod_{\sigma \in {\text{Gal}}(\QQ(\lambda_i)/\QQ)}
		\sigma(y_i)
	\end{aligned}
\end{equation}
is such that the solutions to variables \(\sigma(y_i)\)  define the set of all conjugates of~$\lambda_i$.

Define  the rewrite rules $\tau:=\{\lambda_i \to y_i   \;,
    x_{i,k_i} \to y_i \mid 1\leq i\leq s \}$.
Similarly to the above,
we define the ideal $\mathcal{Y}\in \QQ[\bm{x},y_1,\ldots,y_s]$ obtained by 
 applying the  transformation $\tau(PU^{-1}R^{-1})$ to the vector of variables~$\bm{x}$ for each polynomial in~\eqref{eq-SS},
the set of polynomials given in~\eqref{eq-eign}
 for each~$\lambda_i$, the set $\tau(S_1)$ of multiplicity relations expressed in $y_i$, together with $\tau(P JP^{-1})-M$ and $\tau(U P^{-1})\bm{\alpha} - \bm{\beta}$. 
 
Observe that $V(\mathcal{J})=\pi(V(\mathcal{Y}))$ where 
$\pi$ is the projection map onto the $x_i$ coordinates, 
%see Lemma B.1 in \cite[Appendix B]{manssour2024loops} for a detailed proof.
see \cref{lem:comppolys} in \cref{sec:omitted} for a detailed proof.

\subsection{Computational Complexity}
\label{sec:complexity}

We now proceed with the statement and proof of our main result.

\theoinvpspace*

\begin{proof}
	The construction detailed in~\cref{sec:computestrong}
	can be performed in polynomial time, modulo 
	the computation of the set~$S_1$ of polynomials
	defining $\overline{\{(\lambda_1^n, \ldots, \lambda_i^n, \ldots , \lambda_s^n)\, :\, n \in \N; \, \bm{\beta}_i\neq 0\}}$ in~\eqref{def-s1}.
	
	Indeed, every \emph{rational} square matrix has a Jordan decomposition that can be computed in polynomial time~\cite{cai1994jnf}, and in polynomial time we can shuffle the decomposition so that the Jordan blocks with zero eigenvalues  appear first. 
	Due to the symbolic representation  of eigenvalues~\cite[Section 4.2.1]{cohen2013course}, all constructions in the remaining steps are performed in polynomial time. This includes the computation of~$U$ obtained by applying~\cref{oneinblock} and the computation of its inverse.

	It remains to discuss the computation of~$S_1$.
	Here, we  employ the architecture developed in
	\cite[Lemma 6]{DerksenJK05}. 
	Therein, the problem of computing the 
	polynomial ideal that defines the closure $\overline{\{(\lambda_1^n, \ldots, \lambda_i^n, \ldots, \lambda_s^n)\, :\, n \in \N; \, \bm{\beta}_i\neq 0\}}$ reduces to finding a set of generators 
	for the associated lattice~$L$. 
	To obtain the complexity  bounds, we 
	rely on a theorem of Masser~\cite{masser1988relations} that gives an explicit upper bound  on the magnitude of the components of  a basis for~$L$.
	
	Following~\cite{HosseiniO019}, membership of a tuple $\bm{v} \in L$ is in ${\exists \mathbb{R}}$, using a
decision  procedure for the existential theory of the reals. 
	In combination with Masser's bound, it follows that we can
	compute a basis for~$L$ in $\PSPACE$ by brute-force search if the dimension is not fixed, and in polynomial time when the dimension is fixed.
	See \cref{sec:alg} for further details.
% See Appendix A in \cite{manssour2024loops} for further details.
\end{proof}

In the construction above, the input matrix-vector pair \(\langle M, \bm{\alpha}\rangle\) encodes a loop with rational constants. Our procedures naturally extend to cases where the loop constants are algebraic numbers. This is because   algebraic numbers can be expressed in the power basis of their splitting field, with each number represented as a vector of rational coefficients. This encoding increases the loop's dimension by a factor equal to the degree of the splitting field over~\(\mathbb{Q}\).
For further details about computations with algebraic numbers see \cite{mishra2012algorithmic}.

\section{Worked Examples}
 \label{sec-worked-examples}
 We illustrate the multi-step procedure of \cref{sec-algLoops}
with three worked examples. The computations involved in preparing these examples were performed in \texttt{Macaulay2}~\cite{M2}.

The orbit closure of the loop in \cref{ex:isolated} has three isolated points. By design, this example highlights
the role of the convenient Jordan form and the transformation~$R$ in  computing the ideal defining the invariant. As the update matrix in this example has only a single nonzero eigenvalue, the lattice of multiplicity relations is not relevant.

\begin{example} \label{ex:isolated}
We run our procedure for generating the strongest algebraic invariant for the loop below. 
	\begin{algorithmic}
    \REQUIRE \((x_1,x_2,x_3,x_4,x_5,x_6)\leftarrow  \bluebox{ {$\frac{1}{16}(0,8,14,-5,0,0)$}}\)
		\WHILE{\texttt{true}}
		\STATE \begin{equation*}
			\begin{pmatrix} x_1 \\ x_2 \\ x_3 \\ x_4 \\ x_5 \\ x_6 \end{pmatrix} \leftarrow
			\redbox{ $\displaystyle\frac{1}{14} \begin{pmatrix}
     42  & 0&-7&-42& 21 & 28\\
    -50& 10 &0 &2&-2&-20\\
    -26&-20 & 28 & 52 &-10&-30\\
    -4&-2&0&-6&6&4\\
    14&0&-14&-28&28&14\\
    -38&-12&14&48&-20&-18
\end{pmatrix}$ } 
\begin{pmatrix}
				x_1\\
				x_2\\
				x_3\\
				x_4\\
				x_5\\
				x_6
			\end{pmatrix}
		\end{equation*}
		\ENDWHILE
			\end{algorithmic}
   \captionof{Loop}{Loop defined by the tuple \(\langle \red{M}, \bm{\blue{\alpha}}\rangle\)}
   
\medskip

The tuple $(P,J)$ is a Jordan decomposition  of the update matrix $M$, where
\[
 P = \begin{pmatrix}
    0& -1&0& 0& -1& -1\\
    2& 0& 0& 1& 1& 1\\ 
    1& 1& -1& -1&-1&1\\
       1& 0& 1& 0& 0& 0\\
        1& 1& 1& 1& -1& -1\\
        1& 1& 0& -1& 1& 1
\end{pmatrix}
\enspace \text{and} \enspace
J = \begin{pmatrix}
    0& 1& 0& 0& 0& 0\\
    0& 0& 1& 0& 0& 0\\ 
    0& 0& 0& 0&0&0 \\
       0& 0& 0&2& 1& 0\\
        0& 0&0& 0& 2& 1\\
        0& 0& 0&0& 0& 2
        \end{pmatrix}.
\]
The matrix \(J\) decomposes into two blocks $J= N \oplus J_1$, where the
 block \(N\) is the nilpotent block associated with the eigenvalue~$0$ and  
\(J_1\) is the block associated with the eigenvalue 2.
Since \(N^3=0\), we deduce that $\zski$ has $3$ isolated points $\{ \bm{\alpha}, M \bm{\alpha}, M^2 \bm{\alpha}\}$, see equation \eqref{def-isp}. 
Let 
\[\bm{\gamma} = J^3 P^{-1}\bm{\alpha} =(0,0,0,0,0,1)^\top \qquad \text{and} \qquad \Tilde{J} = \bm{0}_{3\times 3}\oplus J_1.\] 
 This implies that $(\Id_6, \Tilde{J})$ is a convenient Jordan decomposition for $\langle \Tilde{J}, \bm{\gamma}\rangle$; as in \cref{oneinblock} and the discussion at~\eqref{def-isp}.
  For every $n \ge 3$, as in equation \eqref{eq:block}, we have 
\[
M^n \bm{\alpha} = P \Tilde{J}^n \bm{\gamma} = P \begin{pmatrix}
    0\\
    0\\
    0\\
    \tfrac{n(n-1)}{2} 2^{n-2}\\
    n2^{n-1}\\
    2^n
\end{pmatrix}.
\]
Let $R = \Id_3 \oplus \begin{pmatrix}
    8 & 2 & 0 \\
    0 & 2 & 0 \\
    0 & 0 & 1
\end{pmatrix}.$ Then for every $n \ge 3$, as in \eqref{eq:only-powers-n}, we get
\begin{equation} \label{eq:workedexample2}
 M^n \bm{\alpha} = P R^{-1} \begin{pmatrix}
    0\\
    0\\
    0\\
    n^2 2^n\\
    n2^n\\
    2^n
\end{pmatrix}.
\end{equation}
In this simple example, our procedure outputs \(S_{2,1} = \{x_4 x_6 - x_5^2\}\) and \(S_4 =\{x_1, x_2, x_3\}\).  One can also infer these polynomial relations from \eqref{eq:workedexample2}. 
Altogether, the ideal $\langle x_1, x_2, x_3, x_4  x_6 - x_5^2 \rangle$ defines the closure of $ 
 R P^{-1}  \left\{M^n \bm{\alpha} : n \ge 3\right\}$.
After applying the transformation $P R^{-1}$, we obtain the polynomials that define the closure of $\left\{M^n \bm{\alpha} : n \ge 3\right\}$.
Indeed, the transformed relations generate the polynomial 
\begin{equation*}
   I =  \langle x_5+x_6, x_4, 2x_1+x_2+x_6, 4x_2^2+3x_2x_3-3x_3^2-3x_2x_6+x_3x_6-2x_6^2
    \rangle.
\end{equation*}
It follows that the variety defining the orbit closure \(\zski\) of the loop \(\langle M, \bm{\alpha}\rangle\) is given by 
$V(I) \cup\{ \bm{\alpha}, M \bm{\alpha}, M^2 \bm{\alpha}\}$. \hfill $\blacktriangleleft$
\end{example}

In the following example  we consider a loop whose  update matrix  has several nonzero eigenvalues. Thus to compute the orbit closure of the loop, our procedure requires us to first compute a basis for lattice of multiplicity relations between the eigenvalues.

\begin{example} \label{ex:lattice}
We run our procedure on the loop $\langle M, \bm{\alpha}\rangle$, given below. The eigenvalues of \(M\) are all primitive third and fourth roots of unity.

	\begin{algorithmic}
		\REQUIRE \((x_1,x_2,x_3,x_4,x_5,x_6)\leftarrow \bluebox{(2,-1,1,1,0,-1)}\)
		\WHILE{\texttt{true}}
		\STATE \begin{equation*}
			\begin{pmatrix} x_1 \\ x_2 \\ x_3 \\ x_4 \\ x_5 \\ x_6 \end{pmatrix} \leftarrow
			\redbox{ $\begin{pmatrix}
				0 & 0 & 0 & 0 & 0 & -1\\
				1&	0&	0&	0&	0&	-2 \\
				0&	1&	0&	0&	0&	-4 \\
				0&	0&	1&	0&	0&	-4 \\
				0&	0&	0&	1&	0&	-4 \\
				0&	0&	0&	0&	1&	-2
			\end{pmatrix}$  } 
   \begin{pmatrix}
				x_1\\
				x_2\\
				x_3\\
				x_4\\
				x_5\\
				x_6
			\end{pmatrix}
		\end{equation*}
		\ENDWHILE
			\end{algorithmic}
   \label{loop1}
   \captionof{Loop}{Loop defined by the tuple \(\langle \red{M}, \bm{\blue{\alpha}}\rangle\)}

\medskip

Since all the eigenvalues of update matrix \(M\) are nonzero, 
there are  
no isolated points in the orbit closure of the loop. 
Fix $\omega\coloneqq  -\frac{1}{2} + \frac{\sqrt{3}}{2}\iu$. 
The tuple~$(P, J)$ is a Jordan decomposition of ~$M$, where
	\begin{equation*} P =
	\begin{pmatrix}
		-\iu &	\iu &	-\omega &	\omega^2 &	-\omega^2 &	\omega \\
		1-2\iu &	1+2\iu &	\omega^2 + 2 &	-2\omega &	\omega +2 &	-2\omega^2 \\
		2-3\iu &	2+3\iu &	2\omega^2 +3 &	\omega^2 +1 &	2\omega +3&	\omega+1\\
		3-2\iu &	3+2\iu &	\omega^2 + 3 &	-2\omega &	\omega +3 &	-2\omega^2 \\
		2-\iu & 2+\iu &	\omega^2 +2 &	1&	\omega+2&	1\\
		1&	1& 	1&	0&	1&	0
	\end{pmatrix}
\quad \text{and} \quad J = 
	\begin{pmatrix}
		-\iu &	0&	0&	0&	0&	0\\
		0&	\iu &	0&	0&	0&	0\\
		0&	0&	\omega^2&	1&	0&	0\\
		0&	0&	0&	\omega^2&	0&	0\\
		0&	0&	0&	0&\omega&	1\\
		0&	0&	0&	0&	0&\omega
	\end{pmatrix}.
\end{equation*}

The spectrum of $J$ (and hence $M$) is $\{\pm\iu,\omega, \omega^2\}$.
We next transform our decomposition into a convenient one with a change-of-basis matrix (following~\cref{oneinblock}).
We compute an invertible matrix~$U$ for which
$UP^{-1}\bm{\alpha} = (1,1,0,1,0,1)^\top$.
Recall that $U$ admits a block diagonal decomposition \[U \coloneqq  U_1 \oplus U_2 \oplus U_3 \oplus U_4\] where blocks $U_1$ and $U_2$ are size-1 and blocks $U_3$ and $U_4$ are  size-2.
For this computation, we pre-compute the inverse matrix
\[P^{-1} = \frac{1}{18}
\begin{pmatrix} 
-9\iu &	-9 &	9\iu &	9 &	-9\iu &	-9 \\
9\iu &	-9 &	-9\iu &	9&	9\iu&	-9 \\
8(\omega-\omega^2) &	{8-2\omega} & {4(\omega^2-\omega)}  &	{2\omega^2-8} &	{8(\omega-\omega^2)} &	{14-8\omega}\\
6\omega &	6 &	6\omega^2 &	 6\omega &	6 &	6\omega^2 \\
8(\omega^2-\omega) &	8-2\omega^2 & 4(\omega-\omega^2) &	2\omega-8 &	8(\omega^2-\omega) &	14-8\omega^2 \\
6\omega^2 &	6&	6\omega &	6\omega^2 &	6 &	6\omega
\end{pmatrix}
\]
and the vector
$P^{-1}\bm{\alpha} = \alpha_1 \oplus \alpha_2 \oplus \bm{\alpha}_3 \oplus \bm{\alpha}_4$ where 
\begin{align*}
    \alpha_1 &= \frac{3}{2} - \frac{\iu}{2}\, , \qquad &
\alpha_2  = \frac{3}{2} + \frac{\iu}{2}\, , 
\\
\bm{\alpha}_3 & = \frac{1}{9} \left(16\omega-10, 9\omega -3 \right)\, ,  \text{and} &
\bm{\alpha}_4  = \frac{1}{9} \left({16\omega^2 -10}, 9\omega^2 -3  \right).
\end{align*}
By solving the system of linear equations $UP^{-1}\bm{\alpha} = (1,1,0,1,0,1)^\top$ in terms of the entries of \(U\), we obtain  $U_i = \alpha_i^{-1}$ for $i\in \{1,2\}$, and  
\[
U_3 =
\begin{pmatrix}
	\frac{3}{3\omega-1} & \frac{10-16\omega}{(3\omega-1)^2}\\
	0&	\frac{3}{3\omega-1}
\end{pmatrix},\enspace \text{and} \enspace U_4 = 
\begin{pmatrix}
	\frac{3}{3\omega^2-1} & \frac{10-16\omega^2}{(3\omega^2-1)^2}\\
	0&	\frac{3}{3\omega^2-1} 
\end{pmatrix},
\]
which completes the computation of the block matrix \(U\).
From \cref{oneinblock},
we conclude that $\left(PU^{-1}\bm{\alpha}, J\right)$ is a convenient Jordan decomposition for~$\langle M, \bm{\alpha}\rangle$.

Let
$\bm{\beta} \coloneqq  U P^{-1}\bm{\alpha}$.
The next step is to compute a block matrix
\(R = R_1 \oplus R_2 \oplus R_3 \oplus R_4\) such that the blocks of $RJ^n\bm{\beta}$ have the form presented in~\eqref{eq:only-powers-n}.
Since each $R_i$ has size at most 2, the matrix~$R$ is diagonal (see the general form given in~\eqref{eq:Ri}).
Observe that $R = \diag{1,1,\omega^2,1,\omega,1}$.
We now construct a generating set of the ideal~$I_R$ defining the variety \(V(I_R) = R \cdot  \overline{\{ J^n \, \bm{\beta}: n \in \mathbb{N} \}}\), 
as described in \eqref{eq-r-closure}.
We consider each of the sets \(S_1\), \(S_2\), \(S_3\), and \(S_4\) of possible polynomial relations amongst the variables in turn.

First, we consider the set \(S_1\) of multiplicative relations between the eigenvalues. 
This step constructs the exponent lattice $L_{\textrm{exp}}$ of $\{-\iu, \iu, \omega^2, \omega\}$ defined by
\begin{equation*}
	L_{\textrm{exp}} = \left\{(a,b,c,d) \in \ZZ^4 : (-\iu)^a \iu^b (\omega^2)^c \omega^d = 1\right\}.
\end{equation*}
We invoke a standard subroutine to compute a basis for this lattice 
%(see Appendix A in \cite{manssour2024loops} for the details).
(see \cref{sec:alg} for the details).  
An example of such a basis is \(\Lexp{-\iu, \iu, \omega^2, \omega}\) is \(
	\{(4,0,0,0), (1,1,0,0), (0,0,3,0), (0,0,1,1)\}\)
and so we deduce that the polynomials
\[p_{11}(\bm{x}) = x_1^4-1, \quad
p_{12}(\bm{x}) = x_1x_2-1, \quad 
p_{13}(\bm{x}) = x_4^3-1, \quad
p_{14}(\bm{x}) = x_4x_6 - 1\]
define the set~$S_1$.

We next consider each of the sets \(S_{2,i}\) that characterise the polynomial relations amongst the nonzero entries of the $i$th Jordan block.
We note that in this example each set \(S_{2,i}\) is empty because each of the four Jordan blocks has size at most $2$. %
Similarly, the set of polynomial relations $S_4$ is empty. %

Finally, in this example the only non-empty set amongst the $S_{3,i,j}$'s is $S_{3,3,4}$, as the $S_{3,i,j}$ relations require Jordan blocks \(J_i\) and \(J_j\) of size at least~2. The single polynomial relation we infer is \[p_3(\bm{x}) = x_3x_6 - x_4x_5\,.\]

From~\cref{claim-IR}, 
we have $I_R = \langle p_{11}, p_{12},p_{13},p_{14},p_3\rangle$, or
 \[R \cdot  \overline{\{ J^n \, \bm{\beta}: n \in \mathbb{N} \}} = V(I_R).\]

In order to compute the polynomials that define the algebraic set $\zski=\overline{\{M^n \boldsymbol{\alpha} : n \in \mathbb{N} \}}$,
we need to apply the transformation~$PU^{-1}R^{-1}$ to each of the generators of~$I_R$.
For example, we obtain \[p'_{12}(\bm{x}) \coloneqq p_{12}(RUP^{-1}\bm{x}) = \frac{1}{10} (x_1-\iu x_2-x_3+ \iu x_4+x_5-\iu x_6)(x_1+ \iu x_2-x_3- \iu x_4+x_5+ \iu x_6) - 1.\]
The polynomials $p'_{11}, p'_{13}, p'_{14}$, and $p'_3$ are obtained similarly. Thus the algebraic set \(\zski\) is characterised by the variety \(V(\langle p'_{11}, p'_{12}, p'_{13}, p'_{14},p'_3\rangle )\) and we have completed the task of invariant generation, as desired.
\hfill $\blacktriangleleft$
\end{example}

Recall that the degree of the splitting field associated with the eigenvalues of a rational matrix may exhibit exponential growth in the dimension of the  matrix.
In \cref{ex:splitting},
we consider a loop whose linear update is given by the companion matrix of $(x^2-2)(x^2-3)(x^2-5)$.
This matrix has spectrum \(\{\pm\sqrt{2}, \pm\sqrt{3}, \pm\sqrt{5}\}\), and the splitting field of its eigenvalues has degree $2^3$ over $\mathbb{Q}$.
Extending this example, we can construct a matrix with spectrum $\{\pm \sqrt{p_1},\ldots,\pm\sqrt{p_k}\}$ where  $p_i$ is the $i$th prime number. The degree of the splitting field $\mathbb{Q}(\sqrt{p_1},\ldots,\sqrt{p_k})$ over $\mathbb{Q}$ is $2^k$, while the dimension of the matrix is $2k$.

\begin{example} \label{ex:splitting}
We apply our  invariant generation procedure for the following loop.

	\begin{algorithmic}
		\REQUIRE \((x_1,x_2,x_3,x_4,x_5,x_6)\leftarrow \bluebox{(6,0,-62/15,0,2/3,0)}\)
		\WHILE{\texttt{true}}
		\STATE \begin{equation*}
			\begin{pmatrix} x_1 \\ x_2 \\ x_3 \\ x_4 \\ x_5 \\ x_6 \end{pmatrix} \leftarrow
			\redbox{ $\begin{pmatrix}
    0& 0&0&0&0&30\\
    1&0&0&0&0&0\\
    0&1&0&0&0&-31\\
    0&0&1&0&0&0\\
    0&0&0&1&0&10\\
    0&0&0&0&1&0
\end{pmatrix}$    }
\begin{pmatrix}
				x_1\\
				x_2\\
				x_3\\
				x_4\\
				x_5\\
				x_6
			\end{pmatrix}
		\end{equation*}
		\ENDWHILE
			\end{algorithmic}
   \captionof{Loop}{Loop defined by the tuple \(\langle \red{M}, \bm{\blue{\alpha}}\rangle\)}
   
\medskip 

A Jordan decomposition $(P,J)$ of the update matrix $M$ is given by the matrices 
\[
 P = \frac{1}{30} \begin{pmatrix}
    30& 30& 30& 30& 30& 30\\
    15\sqrt{2}& -15\sqrt{2}& 10\sqrt{3}& -10\sqrt{3}& 6\sqrt{5}& -6\sqrt{5}\\ 
    -16& -16& -21& -21&-25& -25\\
       -8\sqrt{2}& 8 \sqrt{2}& -7 \sqrt{3}& 7 \sqrt{3}& -5\sqrt{5}& 5\sqrt{5}\\
        2& 2& 3& 3& 5& 5\\
        \sqrt{2}& -\sqrt{2}& \sqrt{3}& -\sqrt{3}& \sqrt{5}& -\sqrt{5}
\end{pmatrix}
\]
and
\[J = \diag{\sqrt{2}, -\sqrt{2}, \sqrt{3}, -\sqrt{3}, \sqrt{5}, -\sqrt{5}}\,.\]
Since \(P^{-1}\bm{\alpha} = (1,1,1,1,1,1)^\top\), the tuple \((P,J)\) is a convenient form for the pair \(\langle M, \bm{\alpha}\rangle\).
We now compute the polynomials that define the orbit closure \(\zski\).
Since each eigenvalue of \(M\) is simple, the convenient form immediately leads us to 
\[\zski = P \cdot \overline{ \left\{ \bigl((\sqrt{2})^n, \bigl(-\sqrt{2}\bigr)^n, (\sqrt{3})^n, \bigl(-\sqrt{3}\bigr)^n, (\sqrt{5})^n, \bigl(-\sqrt{5}\bigr)^n \bigr)^\top : n \in \N \right\}}.
\]
Then the ideal defining $P^{-1} \zski$ is \[I = \langle x_1^2 - x_2^2 , x_3^2 - x_4^2, x_5 ^2 - x_6^2 \rangle,\] and after applying matrix $P$, we obtain 
\[
\zski =  
V \left(\left\langle \begin{array}{c}
x_3x_4+x_2x_5+10x_4x_5+x_1x_6+10x_3x_6+69x_5x_6,\\
x_2x_3+x_1x_4-31x_4x_5-31x_3x_6-280x_5x_6,\\
x_1x_2+30x_4x_5+30x_3x_6+300x_5x_6
\end{array}\right\rangle\right).
\]
\hfill $\blacktriangleleft$
\end{example}

The complexity of our algorithm for computing the orbit closure of simple linear loops primarily depends on finding a basis for the lattice of multiplicity relations between the eigenvalues of matrices. 
In particular, 
as argued in the proof of \cref{theo-inver-pspace}, all other steps in our algorithm run in polynomial time. 
Ge's algorithm~\cite{ge1993thesis} computes such a basis in polynomial time in the degree of the splitting field of the input algebraic numbers (i.e., the update matrix eigenvalues).
For rational matrices, while the degree of the eigenvalues is bounded by the matrix dimension, the degree of the splitting field can grow exponentially with the dimension.
We observed this in the earlier extension of \cref{ex:splitting} to a family of $2k$-dimensional matrices, where the degree of the splitting field becomes $2^k$. 
Therefore, Ge's algorithm, in its current form, does not allow us to reduce the complexity from {\PSPACE} to polynomial time.

\section{Invariant Verification}
\label{sec-verif}
Recall that the invariant verification problem checks whether the algebraic set defined by a given set of input polynomials is an invariant for an input loop, though it may not be an inductive invariant. The study and motivation behind such non-inductive invariants have been  explored previously~\cite{humenberger2020algebra, kincaid2017reasoning, kovacs2008psolvable, bayarmagnai2024algebraic,  Muller-OlmS04, rodriguez2004inv}.

As discussed in~\cref{sec-overview}, a standard backward algorithm, used in similar settings in~\cite{benedikt2017polynomial,bayarmagnai2024algebraic,kauers2007equivalence,kauers2008solving}, provides  
a conceptually simple procedure for  invariant verification. 
Let   $\langle M,\bm{\alpha}\rangle$ be a simple linear loop. Let 
 $S\subseteq \QQ[\bm{x}]$ be a set of polynomials, written in the dense representation, with the  description size~$s$.
Define the sequence of nested ideals~$(I_i)_{i\in \NN}$ by
	$I_0= \langle S \rangle$ and $
	I_i = \langle P(M^j\boldsymbol{x}): P\in S, j\leq i \rangle$, as defined in~\eqref{eq:chain1}. 
Since $\QQ[\boldsymbol{x}]$ is Noetherian,
this ascending sequence $I_0 \subseteq I_1  \subseteq I_2 \subseteq \cdots $
of  ideals stabilises: there exists~$k$ such that $I_k= I_{k+j}$ for all
$j\in \mathbb{N}$. Denote by $I_{\infty}$ the stabilising value of the
sequence.
The algorithm tests whether $\bm{\alpha} \in  V(I_\infty)$. 
If yes, then $V(S)$ is an algebraic invariant for $\langle M,\bm{\alpha}\rangle$; otherwise,  it is not.

Since $M\boldsymbol{x}$ is a linear transformation, the degree of
$P(M^i\boldsymbol{x})$ is at most the degree of $P(\boldsymbol{x})$. 
From this observation we prove that 
if $\boldsymbol{\alpha}\not\in V(I_{\infty})$ holds then
there exists $k=  O(2^{s})$ such that  
$\boldsymbol{\alpha}\not\in V(I_k)$.
Below, we use this bound to obtain a $\coNP$
upper bound for invariant verification with the input ideal given in dense  representation.
 
\lemweakverifycoNP*
\begin{proof}

Recall that $s$ denotes the description size of~$S$, and
recall 
the chain of nested ideals  $(I_i)_{i\in \NN}$, defined above (and in \eqref{eq:chain1}). 
To show the $\coNP$ membership,
 it suffices to provide a polynomial-time verifiable  certificate  
showing that
 $\boldsymbol{\alpha} \not \in V(I_{\infty})$ for negative instances of the problem. We first prove that
 
\begin{restatable}{claim}{lengthchain}
\label{claim:length}
Suppose that $\boldsymbol{\alpha} \not \in V(I_{\infty})$. There exists $k =O(2^{s})$
such that $\boldsymbol{\alpha} \not \in V(I_{k})$.
\end{restatable}
\begin{proof}
By definition,  the maximal degree~$D$  of polynomials in $S$ and the number~$d$ of variables  are both bounded by~$s$.  
For all $n \in \N$ and all polynomials~$P \in S$,  the degree of  $ P(M^n \boldsymbol{x})$ is at most~$D$, meaning that  
the ideal $I_n$ is generated by polynomials of degree at most $D$.

Let $N\in \NN$ be such that $I_N = I_{N +1}$. This implies that $Q(M^{N +1}\boldsymbol{x}) \in \langle P(M^j\boldsymbol{x}) : P\in S, \, j \leq N \rangle$
for all~$Q\in S$.
But then $Q(M^{N +2}\boldsymbol{x}) \in \{ P(M^j\boldsymbol{x}) : P\in S,\,  j \leq N + 1\} = I_{N +1}=I_N$
for all~$Q\in S$.
Hence, using an inductive argument we can show $I_{N + j} = I_N$ for all~$j \in \N$. 
    Let $k$ be the smallest integer such that $I_{k} = I_{k +1}$, which implies $I_{k} = I_\infty$. Since $I_0 \subsetneq I_1 \subsetneq \cdots \subsetneq I_k$,  at each step $i$ we introduce at least one generator of $I_i$ that can not be expressed as a linear combination of the  generators of $I_{i -1}$. The generators of $I_i$ lie in the vector space of polynomials of degree at most $D$ with dimension~$\binom{D + d}{d}$. Therefore, $k \leq \binom{D + d}{d}$ which completes the proof of the claim.
\end{proof}

This claim implies that 
  there exists $P\in S$ such that 
  $P(M^k\boldsymbol{x}) (\boldsymbol{\alpha})\neq 0$.  
A $\coNP$ algorithm  guesses $P(\bm{x})\in S$,  an index $k =O(2^{s})$  and a prime $p$
with bitsize~$s$ such that $P(M^k\boldsymbol{x}) (\boldsymbol{\alpha})\neq 0$.
Taking the binary representation of~$k$ into account and, by standard doubling techniques, 
the algorithm constructs a small circuit 
for~$P(M^k\boldsymbol{x})$  and checks whether $P(M^k\boldsymbol{x})\not \equiv 0$ (mod~$p$) in polynomial time.
We borrow the correctness of the latter test from the well-known fingerprinting procedure for the ACIT problem~\cite{allender2009complexity}.

The proof of {\coNP}-hardness is by a  reduction from  \(\mathsf{3SAT}\)  to the complement of the invariant verification problem. 
Given a  \(\mathsf{3SAT}\) formula~$\Phi$, we construct a  loop~$\langle M, \bm{\alpha} \rangle$ with orbit~$\orbit$  and a polynomial~$Q(\bm{x})$ such that 
$\Phi$ is satisfiable if and only if $\orbit \not \subseteq V(\langle Q \rangle)$.

Let $\Phi = \bigwedge_{i=1}^m C_m$ be in CNF  over 
variables~$\{y_1, \dots, y_k\}$. Let $p_1 <  \dots < p_k$ be the first~$k$ primes. Define $D(i) \coloneqq  1+\sum_{j=1}^{i-1} p_j$ for $i\in \{1,\ldots,k\}$ and 
 $d\coloneqq  \sum_{j=1}^{k} p_j$. 
Construct $M \in \{0,1\}^{d\times d}$ and~$\bm{\alpha} \in \{0,1\}^d$ as follows:
 the entry $\alpha_{\ell}$  of $\bm{\alpha}$  is $1$ if, and only if, 
	$\ell = D(i)$ for some $i \in \{1, \dots, k\}$. 
The matrix~$M$ is a 
 block-diagonal matrix $ M_1 \oplus \dots \oplus M_k$,
	where $M_j$ is a permutation matrix of size~$p_j$.
Intuitively speaking, for any $n \in \NN$ the vector $M^n\bm{\alpha}$ splits into $k$ blocks of prime size, and  exactly one entry in each block is set to~$1$. By the Chinese remainder theorem, for every $\bm{\beta} \in \{0,1\}^d$ with exactly one~1 entry in each block,
	there exists~$n\in \NN$ such that $\bm{\beta} = M^n\bm{\alpha}$.

 We construct the polynomial~$Q$ over the vector $\bm{x}=(x_1,\cdots, x_d)$ of variables. 
 For each $1 \leq i \leq k$, we interpret the variables $x_{D(i)}$ and $x_{D(i)+1}$ as 
literals $y_i$ and $\neg y_i$. The other variables are called \emph{non-literal}.
Define $Q\in \QQ[\bm{x}]$ as follows 
	\[Q (\bm{x}) \coloneqq  \left(\prod_{x_i\text{ non-literal}} (1-x_i)\right) \cdot \prod_{i=1}^m Q_i(\bm{x}),\] 
	where $Q_i \coloneqq  \sum_{j=1}^k t_{ij}^2$ with 
\[t_{ij} = \begin{cases}
		x_{D(j)+1} & \text{if $y_j$ appears in $C_i$,}\\
		x_{D(j)+2} & \text{if $\neg y_j$ appears in $C_i$,}\\
		0 & \text{otherwise.}
	\end{cases}\]

The correctness of reduction follows from   two simple observations:  
$Q(\bm{x})$ vanishes on~$\bm{\beta}\in \{0,1\}^{d}$  if $\beta_i = 1$ for some non-literal entry $i$. Furthermore,  
provided that  all  non-literal entries are zero, the point~$\bm{\beta}$ is a zero of $Q(\bm{x})$ if and only if it corresponds to an unsatisfying assignment of~$\Phi$. 
\end{proof}

The backward algorithm employed in \cref{lem-weakverify-coNP} results in an inefficient \(\EXPSPACE\) upper bound  when the input polynomials have sparse representation. 
Here our route to obtaining a tighter complexity bound is through invariant generation, specifically computing the strongest (inductive) algebraic invariant \(\zski\).
By \cref{theo-inver-pspace}, we can construct a  set $\{I_1,\ldots,I_k\}$
of ideals  such that~$\zski=V(I_1 \cap \cdots \cap I_k)$.
The generating set of each ideal~$I_i$ has a small cardinality, and  comprises of    polynomials in $\QQ[\boldsymbol{x}]$ written in sparse representation with size polynomial in the  input loop description. We  perform several radical membership tests to determine the relation of~$\zski$ with~$V(S)$.

\lemweakverify*

\begin{proof}
Let
$\langle M,\bm{\alpha}\rangle$ be the input loop with orbit~$\orbit$.
	With \cref{theo-inver-pspace} at our disposal,
	we construct the  list of generators for the ideals 
    {$\mathcal{I}_0, \ldots,\mathcal{I}_{n_0-1}\subseteq \QQ[\bm{x}]$} and {$\mathcal{Y}\subseteq  \QQ[\bm{x}, y_1, \dots, y_s]$}.
Our construction is such that, for all~$i \in \{0,\ldots,n_0-1\}$,  we have
\(
\{M^i \bm{\alpha} \}=V(\mathcal{I}_i)
\). Moreover,
\(
(\bm{v}, \lambda_1,\ldots, \lambda_s) \in V(\mathcal{Y}) 
\) 	 
if and only if~$\bm{v}\in \overline{\{M^n \boldsymbol{\alpha} : n \geq n_0 \}}$.

We now demonstrate how to check in $\PSPACE$
	 whether 
	 \[V(\mathcal{I}_0 \cap \cdots \cap \mathcal{I}_{n_0-1} \cap \mathcal{Y}) \sim V(S)\] 
	 holds for  $\sim  \in \{=, \subseteq \}$.	
Since $\mathcal{Y}$ is defined with extra variables $y_i$'s, we  add  the relations  
between these extra variables to~$S$. 
(More precisely, for each $\lambda_i$, $1 \leq i \leq s$,   we add the set of polynomials defined in~\eqref{eq-eign}, together with $\tau(S_1)$, $\tau(P JP^{-1})-M$ and $\tau(U P^{-1})\bm{\alpha} - \bm{\beta}$ used in definition of $\mathcal{Y}$ from $\mathcal{J}$.)
We first explain how to check $\overline{\orbit} \subseteq V(S)$ in $\PSPACE$.
Towards this, we first check if the rational points  $M^{i}\bm{\alpha}$ with 
$i\in \{0,\ldots,n_0-1\}$ lie in $V(S)$,
this   reduces to the ACIT and can be tested in randomised polynomial time~\cite{allender2009complexity}. 
It remains to check whether $V(\mathcal{J}) \subseteq V(S)$; 
we verify this through testing whether $P \in \sqrt{\mathcal{J}}$
for each polynomial $P\in S$, which  
 reduces to radical membership testing.
 As explained in~\cref{sec-overview}, the latter task is in $\AM$
 under GRH and in $\PSPACE$ unconditionally. 

To conclude, we explain how to test $V(S) \subseteq \overline{\orbit}$ in $\PSPACE$.
We test whether every polynomial~$P\in \mathcal{I}_0 \cap \cdots \cap \mathcal{I}_{n_0-1} \cap \mathcal{J}$ is a member of~$\sqrt{S}$.
This test could be algorithmically expensive due to the intersection, as  a generating set for the intersection of these ideals can face an exponential blow-up in size. We instead test the complement of this question, by guessing one generator from each $\mathcal{I}_{i}$ and 
one from $\mathcal{J}$ and checking whether the product is not a member of~{$\sqrt{S}$}. Here again we
rely on  the $\AM$ protocol of 
radical membership testing.

The hardness follows from the construction given in~\cref{lem-weakverify-coNP}.
\end{proof}

\section{Bit-bounded Synthesis}
\label{sec-bitbounded}

The results in this section are motivated by the HTP-hardness results in \cref{claim:undec}.
Herein we consider complexity bounds for variants of the loop synthesis problem where one bounds the bitsizes of the loop components (\cref{th:bitbounded,th:bitboundedfix}).
In \cref{th:bitboundedfix}, we obtain tighter bounds by placing an additional dimension specification on the loop.

\thbitbounded*

\begin{proof}
The $\PSPACE$ bound follows by guessing~$M$ and $\boldsymbol{\alpha}$ with entries respecting the required bit bounds, and using the 
 invariant verification subroutine in \Cref{lem-weakverify} %(\Cref{sec-verif}) 
 applied to~$\langle M, \boldsymbol{\alpha} \rangle$ and the input ideal.
Thus all that remains is to prove the claimed lower bounds. 

Our lower bounds  are obtained by
 reductions from $\mathsf{3SAT}$,
and $\mathsf{Unique\;3SAT}$, following   the folklore encoding of $\mathsf{3SAT}$ in HN~\cite{koiran1996hn}.
The following encoding is the main building block for both 
reductions.  Let~$\Phi := \bigwedge_{i=1}^m C_m$ be in CNF  over 
variables~$\{y_1, \dots, y_d\}$.
From~$\Phi$ we construct a set of polynomials $S\subseteq \QQ[\bm{x}]$ with  vector~$\bm{x}=(x_1, \dots, x_{d+2})$  of variables. We initialise~$S$ with polynomial~$x_{d+1}-x_{d+2}$.
For each boolean variable $y_i$, with~$i\in \{1,\ldots,d\}$, we add the polynomial~$x_i(1-x_i)$ to~$S$.
For each clause~$C_j$, with~$j\in \{1,\ldots,m\}$, we add $P_i := \prod_{j=1}^d t_{ij}$ to~$S$, where  \begin{equation}\label{encode-lit}
	t_{ij} = \begin{cases}
		1-x_j & \text{if $y_j$ appears in $C_i$,}\\
		x_j & \text{if $\neg y_j$ appears in $C_i$,}\\
		1 & \text{otherwise.}
	\end{cases}
\end{equation}

%%%%%%%%
\medskip 

\noindent \textbf{Weak bit-bounded synthesis over \(\{\Q,\Z\}\) is \(\NP\)-hard:}
The proof is by a reduction from~\(\mathsf{3SAT}\). Given an instance~$\Phi$ of~\(\mathsf{3SAT}\),  construct the set $S$ of polynomials, as described above. 

Assume that $\Phi$ is satisfiable. Given  a satisfying assignment, 
define $\bm{\alpha}$ such that  $\alpha_i=1$ if and only if $y_i$ is true in the assignment, and set $\alpha_{d+1}=\alpha_{d+2}=1$. 
Clearly, the infinite orbit of the point $\bm{\alpha}$ under the matrix~$M:=\diag{1, \dots, 1, 2, 2}$
lies in $V(\langle S \rangle)$.
Conversely, if $\Phi$ is unsatisfiable, the  set $V(\langle S \rangle)$ is empty
and, a fortiori, no loop exists for~$\langle S \rangle$.	

We note in passing that that the entries of such $M$ and $\bm{\alpha}$ have constant bitsize.

\medskip

\noindent {\textbf{Strong bit-bounded  synthesis is \(\NP\)-hard under randomised reductions:}}
Recall that $\mathsf{Unique\;3SAT}$ is 
 \(\NP\)-hard  under randomised polynomial-time reductions~\cite{valiant1986unique}. It is 
a promised version of
$\mathsf{3SAT}$, where the  input formula is 
\emph{promised}
to have \emph{at most~one} satisfying assignment.

 Given an instance~$\Phi$ of~\(\mathsf{Unique\;3SAT}\), construct again the set $S$ of polynomials, as described above.
By the promise on~$\Phi$, the projection of $V(\langle S\rangle)$ into the first $d$ coordinates is either empty or a singleton.
The proof is immediate from the previous case, and 
by the observation that 
if $\Phi$ is satisfiable, the projection of~$\overline{\{ M^n \bm{\alpha} : n \in \mathbb{N} \}}$ into the last two coordinates is exactly  $V(\langle x_{d+1}-x_{d+2}\rangle)$. 
\end{proof}

We now improve the complexity lower bounds under an additional dimensionality assumption.

\thbitboundedfix*

\begin{proof}

The $\NP$ bound follows by guessing~$M$ and $\boldsymbol{\alpha}$ with entries respecting the required bit bounds, and using the 
 invariant verification subroutine in  \Cref{lem-weakverify}  applied to~$\langle M, \boldsymbol{\alpha} \rangle$ and the input ideal.

The {\NP} lower bound is by reduction from the quadratic Diophantine equations problem,  known to be {\NP}-complete~\cite{manders1978quadratics,garey1979computers}:
it asks, given natural numbers~$a, b, c$,  
whether there a solution \((x,y)\in\N^2\) to the equation $ax^2+by=c$.
Given an instance~$(a,b,c)$ of the 
 quadratic Diophantine equations problem~$(a,b,c)$, 
	we construct polynomials $L, P, Q \in \ZZ[\bm{x}]$
 with vector~$\bm{x}=(x, y, y_1,y_2, y_3, y_4, z_1, z_2)$ of variables 
	such that there is an integer point in the variety $V(\langle L, P \rangle ) \subseteq \alg^6$
	if and only if the original equation $ax^2+by=c$ has a solution \((x,y)\in\N^2\).
	Define
	\begin{align*}
		L (\bm{x}) &:= y - y_1^2-y_2^2-y_3^2-y_4^2, \\
	P(\bm{x}) &:= ax^2+by-c, \\
	Q(\bm{x}) &:= z_1 - z_2.
	\end{align*}

	By Lagrange's four-squares theorem, every positive integer~$y$ can be expressed as a sum of four integer squares; %
    thus polynomial \(L\) ensures that variable \(y\) can attain only non-negative integer values.
Suppose the intersection $V(\langle L, P, Q\rangle ) \cap \ZZ^8$ is non-empty.
Observe that for \(\bm{\alpha}\in V(\langle L, P, Q\rangle ) \cap \ZZ^8\),
the orbit under $M:=\diag{1, \dots, 1, 2, 2}$ is infinite, and thus$\langle M, \bm{\alpha} \rangle$	 is a non-trivial loop.  
The converse direction is immediate.
\end{proof}

\section{Further Discussion}
\label{sec-conclusion}
We suggest several directions for further research inspired by our contributions to invariant generation for simple linear loops and  loop synthesis presented herein.

\medskip 

\noindent {\bf Invariant Generation for  affine programs.}
A program is considered affine if it exclusively features nondeterministic branching (as opposed to conditional branching) and all its assignments are defined by affine expressions. The 
invariant generation problem for 
affine programs is  addressed by  the
algorithm in \cite{hrushovski2023strong} 
 through the group-closure problem. This problem entails computing a generating set of polynomials for the  Zariski closure $\overline{\langle M_1, \ldots, M_k \rangle}$ for a given  set~$\{M_1, \ldots, M_k\}$ of invertible rational matrices.
 The tightest complexity bound for solving the group-closure problem is severalfold exponential time~\cite{NPSHW2021}.

The main result in this paper, \cref{theo-inver-pspace},  presents a  
{\PSPACE} algorithm for generating the invariants of a simple linear loop (the class corresponding to branch-free loops with a single linear update). To extend our technique to the general case of affine programs, we may first consider  the setting with multiple linear updates~$M_1, \ldots, M_k$ where the matrices are  commutative and invertible. 
Since the $M_i$ commute,  the orbit of the loop is
defined by~$\orbit = \{ M_1^{n_1} \cdots M_k^{n_k} \boldsymbol{\alpha}$ : $n_1, \ldots, n_k \in \N\}$ where $\bm{\alpha}$ is the initial vector.

Recall
that a matrix \(M\in \QQbar^{d\times d}\) is {\emph{unipotent} if there exists \(n\in\N\) such that \((M -\Id_d)^n=0_d\) (here \(\Id_d\) and \(0_d\) are the $d\times d$ identity and zero matrices, respectively)} and $M$ is
\emph{semisimple} if it is diagonalisable over \(\QQbar\).
Define $G\coloneqq  \overline{\langle M_1, \ldots, M_k \rangle}$ so that $\zski = \overline{ G  \boldsymbol{\alpha}}$. 
It is known that the subset of semisimple matrices in~$G$,
denoted by $G_s$, forms an algebraic subgroup; {likewise the set of unipotent matrices in~$G$, denoted
by $G_u$, forms an algebraic subgroup. 
By the Jordan--Chevalley decomposition, we have~$\zski = \overline{G_u G_s \boldsymbol{\alpha}}$.}

In the case that \(G=G_s\) we have the following.
\begin{lemma}
Let \(G\) 
be a  semisimple commutative group. A set of polynomials defining~$\overline{ G  \boldsymbol{\alpha}}$ is computable in {\PSPACE}. 
\end{lemma}
\begin{proof}
    Since $G$ is a semisimple commutative group, 
    there exists $P \in \GL_d(\QQbar)$ such that for all~$i \in \{ 1, \ldots, k\}$, the matrix~$D_i \coloneqq P^{-1} M_i P $ is  diagonal. Following~\cref{oneinblock}, we can choose $P$ in such a way that $P^{-1} \boldsymbol{\alpha} \in \{0,1\}^d$. 
    Writing $\zski = P \overline{ \langle D_1, \ldots, D_k \rangle P^{-1}\boldsymbol{\alpha}}$,
    we observe that $\overline{ \langle D_1, \ldots, D_k \rangle P^{-1}\boldsymbol{\alpha}}$ is the 
    closure of a group of diagonal matrices of dimension at most \(d\).
    To obtain a {\PSPACE} procedure it suffices to closely follow our construction in~\cref{theo-inver-pspace}. The sole difference from that construction lies in the set
 $S_1$ of polynomial equations, which  is now defined  by the intersection of lattices $L_i$ of multiplicity relations between the entries   
 of~$D_i P^{-1}\boldsymbol{\alpha}$~\cite[Chapter 3]{bombieri2006heights}.
\end{proof}

A  natural direction for future research  involves extending our {\PSPACE} procedure to apply to commutative matrices more broadly. 
This requires a better understanding of the polynomial map
\begin{align*}
    G_u \times \overline{ G_s \boldsymbol{\alpha}} &\to \overline{ G\boldsymbol{\alpha}},\\
    (g, \boldsymbol{v}) &\mapsto g \boldsymbol{v}.
\end{align*}
Subsequently, an ambitious objective is to develop a procedure with improved complexity bounds for the  invariant generation problem in affine programs.

\medskip 

\noindent {\bf Loop Synthesis.}
Our results for bit-bounded synthesis over the rationals (\cref{th:bitbounded,th:bitboundedfix}) demonstrate an inherent source of hardness (Hilbert's Tenth problem).
A first direction for future research might consider circumventing such obstacles
by focusing on classes of ideals with an abundance of rational solutions.  A small initial step in this direction, the synthesis of loops for pure-difference binomials, was shown in \cite{kenison2023polynomial}.
An ultimate goal is the synthesis of loops for the larger class of binomial ideals.

\newpage
\begin{acks}
The authors thank James Worrell for his valuable comments and feedback. G.~Kenison and A.~Varonka are grateful for their financially supported travel (UKRI Frontier Research Grant EP/X033813/1). M.~Shirmohammadi and R.~Ait El Manssour are supported by the International Emerging Actions grant (IEA’22), and by ANR grant VeSyAM (ANR-22-CE48-0005). %

\includegraphics[height=1em]{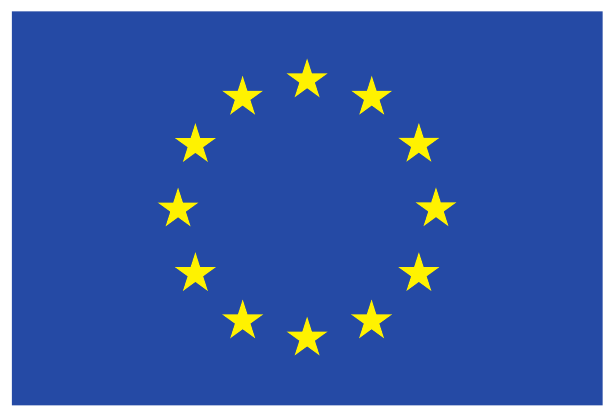} \quad This paper is part of a project that has received funding from the European Research Council (ERC) under the European Union's Horizon 2020 research and innovation program (grant agreement No.~10103444).
A.~Varonka gratefully acknowledges the support of the ERC consolidator grant ARTIST 101002685 and the Vienna Science and Technology Fund (WWTF) 10.47379/ICT19018.

We thank the hosts of Autob\'{o}z 2023 in Kassel, Germany for providing the space where initial discussions for this paper took place.

\end{acks}

%\newpage
\printbibliography

\appendix
\onecolumn
\section{Extended Background}
\label{sec-app-back}
\subsection{Complexity Theory}
An \emph{Arthur--Merlin protocol} is a two-message interactive proof system
between two participants, respectively called Arthur and Merlin.
Arthur can be considered as a probabilistic polynomial time verifier,
while Merlin is an oracle with infinite computational power.  Given an
instance of the problem at hand, Arthur generates some random bits and
sends the bits (together with the instance) to Merlin.  Merlin then
replies to Arthur with a message of length polynomial in the instance.
Arthur decides to accept or reject based on the random bits and
Merlin's message.

The complexity class \AM~of problems that admit Arthur--Merlin protocols 
was introduced by Babai~\cite{Babai1985AM,babi1988arthurmerlin}, motivated by decision problems about matrix groups.
A problem is solved by such a protocol if, for a ``yes\rq\rq\ instance, the
probability over Arthur's bits that there exists a response by Merlin
that causes Arthur to accept is $\geq 2/3$, while for a ``no\rq\rq\ instance,
the probability is $\leq1/3$.  

The complexity class \BPP~is the class of decision problems solvable by an \NP~machine such that \begin{itemize}
	\item a ``yes\rq\rq\ instance is accepted with probability at least~2/3,
	\item a ``no\rq\rq\ instance is accepted with probability at most~1/3,
\end{itemize}
Roughly speaking, an \AM~protocol can be viewed as a \BPP~machine that makes one call to an \NP~oracle accepting its answer.

Important to us is the following chain of inclusions: \(\BPP \subseteq \AM \subseteq \PH \subseteq \PSPACE\).
(Here, \PH~denotes the union of all classes in the \emph{polynomial hierarchy}.)

%%%

\subsection{Ideals and Varieties}
Let $R$ be a ring.
A \emph{polynomial ideal} is a subset $I \subseteq R[\bm{x}]$ that satisfies the following properties:
$0\in I$; \(I\) is closed under addition; and for each $p \in R[\bm{x}]$ and $q \in I$, necessarily $pq \in I$.
For a set~$S \subseteq R[\bm{x}]$  of polynomials, the \emph{ideal generated by $S$} is given by
\(
	I = \ideal{S} := \{ s_1 q_1 + \cdots + s_\ell q_\ell : s_j \in S, q_j \in R[\bm{x}], \ell \in \N \}\).
A polynomial ideal \(I\) is \emph{proper} if \(I\) is not equal to \(R[\bm{x}]\)
and \(I\) is \emph{radical} if \(p^n \in I\) implies that \(p\in I\).
The \emph{radical} $\sqrt{I}$ of an ideal~$I$ is defined by 
\(\sqrt{I} = \{p \in R[\bm{x}] : p^n \in I \text{ for some } n\in \NN\}\).
The radical~$\sqrt{I}$ is an ideal of~$R[\bm{x}]$ itself.

A \emph{Noetherian} ring is a ring~$R$ that satisfies an \emph{ascending chain condition} for ideals.
That is, every chain of inclusions $I_1 \subseteq I_2 \subseteq \cdots$ 
where each $I_j$ is an ideal of~$R[\bm{x}]$
\emph{stabilises}, i.e., there exists an ideal $I_n$ such that $I_n = I_{n+1} = \cdots$.
We recall that~$R \in \{\ZZ, \QQ, \alg\}$ are Noetherian.
Henceforth, the rings we discuss will always be assumed to be Noetherian.

Hilbert's Basis Theorem states that every ideal in \(R[\bm{x}]\) has a finite basis. Seminal work by Buchberger introduced Gr\"{o}bner bases for polynomial ideals, which permit the algorithmic computation of key properties of polynomial ideals~\cite{buchberger2006jsc,cox2015ideals}, including ideal membership, ideal union/intersection, elimination ideals, and many more. 

Hilbert's Nullstellensatz is a celebrated theorem that identifies that radical ideals are exactly those ideals that correspond to varieties.
Formally, let $I$ be an ideal in~$\alg[\bm{x}]$.
A \emph{variety} is the set of common zeros of a polynomial ideal~$I$ so that \(V(I):=\{\bm{v} \in \alg^d : f(\bm{v})=0 \text{ for all } f\in I\}\).
Hilbert's Nullstellensatz states that the ideal of all polynomials in~$\alg[\bm{x}]$ that vanish on $V(I)$ is $\sqrt{I}$.

A variety $V\subseteq \alg^d$ is \emph{irreducible} if it cannot be written 
as $V = V_1 \cup V_2$ such that
$V_1$ and $V_2$ are both varieties properly contained in~$V$. 
The dimension of a variety~$V$ is defined to be the maximum number~$n\in \NN$ such that 
there is a strictly increasing chain $V_0 \subset V_1 \subset \cdots \subset V_n$ of non-empty irreducible subvarieties of~$V$. 
A variety~$V \subseteq \alg^d$ has dimension at most~$d$.

The \emph{Zariski topology} on $\alg^d$ is the topology whose closed sets are varieties.
The \emph{Zariski closure}~$\overline{X}$ of a set~$X \subseteq \alg^d$ 
is the smallest variety that contains~$X$.

\subsection{Multiplicative Relations Between Algebraic Numbers}
\label{sec:alg}
Let us first recall some standard terminology.
Let \(M\) be a \(d\times d\) square matrix with entries in \(\Q\) (or \(\alg\)).
Then the distinct characteristic roots \(\lambda_1,\ldots, \lambda_s\) of \(M\) are algebraic numbers.
To every vector $\bm{v} \in \ZZ^s$ we associate a unique \emph{canonical binomial}, $P(y_1, \dots, y_s) := \bm{y}^{{\bm{v}_+}} - \bm{y}^{{\bm{v}_-}}$, where 
$\bm{v}_+ =(\max\{v_1, 0\}, \dots, \max\{v_s, 0\}) \in \NN^s$
and $\bm{v}_- = \bm{v}_+ - \bm{v} \in \NN^s$. 
Here, $\bm{y} = (y_1, \dots, y_s)$ and $\bm{y}^{\bm{\alpha}}$ denotes $y_1^{\alpha_1}\cdots y_s^{\alpha_s}$ for a vector~$\bm{\alpha} = (\alpha_1, \dots, \alpha_s)\in\NN^s$.

In the course of our work, we need to compute the relations between the 
geometric sequences
\(\langle \lambda_1^n\rangle_n, \ldots, \langle\lambda_s^n \rangle_n\).
These relations are, in turn, rooted in the set of multiplicative relations of nonzero algebraic numbers~$\lambda_1, \ldots, \lambda_s$.
Recall that $(n_1, \ldots, n_s) \in \ZZ^s$ is a \emph{multiplicative relation} of $(\lambda_1, \dots, \lambda_s) \in \alg^s$ if $\lambda_1^{n_1}\cdots \lambda_s^{n_s} = 1$.
For the aforementioned set of multiplicative relations, the \emph{exponent lattice} of $\lambda_1, \dots, \lambda_s \in \alg$, is thus defined as follows:
	\begin{equation*} \textstyle
		\Lexp{\lambda_1, \dots, \lambda_s} := \left\{(n_1, \dots, n_s) \in \ZZ^s : \prod_{i=1}^s \lambda_i^{n_i} = 1\right\}.
		\end{equation*}
We note the set of all multiplicative relations forms an additive free abelian group (\emph{lattice}).

The \emph{height} of $\lambda \in \alg$ is the maximum absolute value 
of the coefficients of its minimal polynomial.

\begin{theorem}[Masser]
	Let $\lambda_1, \dots, \lambda_s \in \alg$.
	The exponent lattice~$\Lexp{\lambda_1, \dots, \lambda_s}$ has 
	a basis $\bm{v_1}, \dots, \bm{v_r}$ for which
	\[\max_{i,j} |v_{i,j}| \leq \left( D \log{H} \right)^{O(s^2)}, \]
	where $H$ and $D$ are upper bounds for heights and degrees of $\lambda_i$'s, respectively.
\end{theorem}

In order to determine the dependencies between the geometric sequences
\(\langle \lambda_1^n\rangle_n, \ldots, \langle\lambda_s^n \rangle_n\),
we first compute a basis \(\bm{v}_1,\ldots, \bm{v}_r\) such that \(\Lexp{\lambda_1,\ldots, \lambda_s} = \Z\bm{v}_1 + \cdots + \Z\bm{v}_r\),
employing Masser's bound.
As a second step, we introduce variables $y_1, \dots, y_s$ to represent 
the geometric sequences.
We further read a set of canonical binomials $P_1, \dots, P_r \in \ZZ[y_1,\dots,y_s]$ from vectors \(\bm{v}_1,\ldots, \bm{v}_r\).

Under the assumption $\lambda_i \neq 0$ for all $1 \leq i \leq s$, we have that
\[\overline{V(P_1, \dots, P_r) \setminus V\left({\textstyle \prod_{i=1}^s y_i}\right)} = \overline{\{(\lambda_1^n, \dots, \lambda_s^n) : n \in \NN\}}.\]
The proof can be found in~\cite[Lemma 6]{DerksenJK05}.
The following is corollary is then immediate.
\begin{corollary}\label{eig-depend}
	Let~$\lambda_1, \dots, \lambda_s$ be nonzero eigenvalues of a rational matrix.
	The Zariski closure of the set 
	\(\{(\lambda_1^n, \dots, \lambda_s^n) : n \in \NN\}\)
	is computable in \PSPACE.
\end{corollary}
%%%
%%%%
%%%
\section{Omitted Proofs}
\label{sec:omitted}
This section details the proofs omitted from the main text. 
The first is a hardness result omitted from \cref{sec-overview}.
\claimundec*
\begin{proof}

	Let $P_1, \dots, P_k \in \ZZ[x_1, \dots, x_d]$  define an instance of HTP over~$R\in\{\Q,\Z\}$. 
	We construct an instance of the weak synthesis problem in $d+2$ variables as follows. Define the input ideal  $\mathcal{S}=\langle P_1, \dots, P_k, x_{d+1} - x_{d+2} \rangle$ in $\ZZ[x_1, \dots, x_d, x_{d+1},x_{d+2}]$.

	Clearly, negative instances (no  solutions in $R^{d}$ for the system of equalities $P_i=0$) remain negative, because polynomials $P_1, \dots, P_k$ still have to vanish on every valuation of the loop.
	Now assume that the original instance is positive, and so there exists a solution to the system of equations.
	Let $(\xi_1, \dots, \xi_d) \in R^d$ be a solution.
	Then there exists a non-trivial linear loop for the produced instance of the Weak Synthesis Problem for loops over~$R$, namely one with an initial vector $\bm{\alpha} = (\xi_1, \dots, \xi_d, 1, 1)$,
	and an update matrix~$\diag{1, \dots, 1, 2, 2}$. 
	The constructed loop is non-trivial,  concluding the proof of the reduction correctness.
\end{proof}

Moving on, as an extra step in \cref{sec:computestrong}, we
generate a set of polynomials with rational coefficients that define the strongest algebraic invariant of a loop \(\langle M, \bm{\alpha}\rangle\). 
This extra step introduces additional variables and we require a technical argument to settle equality between two varieties defined in terms of the polynomial ideals \(\mathcal{J}\subseteq \mathbb{Q}(\lambda_1,\ldots, \lambda_s)[\bm{x}]\) and \(\mathcal{Y}\subseteq\mathbb{Q}[\bm{x}, y_1,\ldots, y_s] \).
The definitions for \(\mathcal{J}\) and \(\mathcal{Y}\) are given in \cref{sec:computestrong}.
\begin{lemma} Let $\pi$ be the projection map onto the $x_i$ coordinates. We have $V(\mathcal{J})=\pi(V(\mathcal{Y}))$. \label{lem:comppolys}
\end{lemma}
\begin{proof}
Denote by $\bm{\lambda}=(\lambda_1, \ldots, \lambda_s)$ the nonzero eigenvalues of $M$  read from its convenient Jordan form~$(P,J)$, as used in our invariant generation construction.

Let $V=V(\mathcal{J})$ and $W=V(\mathcal{Y})$. By construction of $\mathcal{Y}$, the inclusions 
\[V \times \{(\lambda_1,\ldots,\lambda_s)\}\subseteq W\subseteq \alg^d\times \{ (\omega_1, \ldots,\omega_s) \mid m_{\lambda_i}(\omega_i )=0 \}\]
hold.  Clearly, we also have
$V \subseteq \pi(W)$. Below we consider the converse inclusion.

Let $(\bm{v},\bm{\omega}) \in W$ where $\bm{v}\in \alg^d$ and $\bm{\omega}=(\omega_1, \ldots,\omega_s)$. 
As  seen in the above inclusions, for each~$i$, $\omega_i$
is a Galois conjugate of~$\lambda_i$, and moreover  
each conjugate appears with equal multiplicity in~$\bm{\lambda}$ and~\(\bm{\omega}\).  
These ensure that the symbolic construction of $P^{-1}$ given in~\cite{cai1994jnf} is such that 
$\tau(P)$ and $\tau(P^{-1})$ remain inverse when evaluated at $\bm{\omega}$ (rather than $\bm{\lambda}$).
As a result, \(\tau(J)\) evaluated at \(\bm{\omega}\) is equal to \(J\), up to the ordering of Jordan blocks.
 Furthermore, 
the equations in $\tau(S_1)$ ensure that \(\bm{\omega}\) satisfy the same multiplicative relations as those among $\bm{\lambda}$.

Define the rewrite rules $\tau':=\{\lambda_i\to \omega_i \mid 1\leq i\leq s\}$.
The equations  \(\tau(P J P)^{-1} - M\) and \(\tau(U P^{-1}) \bm{\alpha} - \bm{\beta}\) ensure analogous Jordan block configurations and fingerprint for $(\tau'(P),\tau'(J))$ as in $(P,J)$. 
By the above, and our invariant generation construction, we get that
            $V = V(\tau'(\mathcal{J}))$.
         It immediately follows that \(\bm{v}\in V\).
        Since we began with the assumption \((\bm{v} ,\bm{\omega})\in W\), the desired inclusion that \(\pi(W)\subseteq V\) follows.
\end{proof}

\end{document}